\documentclass[letterpaper,USenglish,cleveref,autoref,thm-restate]{lipics-v2021}

\pdfoutput=1 %
\hideLIPIcs

\title{Efficient and Stable Multi-Dimensional Kolmogorov-Smirnov Distance} %

\author{Peter Matthew Jacobs}{University of Wisconsin--Madison, USA \and \url{https://jacobs269.github.io} }{}{}{}

\author{Foad Namjoo}{University of Utah, USA \and \url{https://users.cs.utah.edu/~foad27/} }{foad.namjoo@utah.edu}{}{}%

\author{Jeff M. Phillips}{University of Utah, USA \and \url{https://users.cs.utah.edu/~jeffp}}{jeffp@cs.utah.edu}{https://orcid.org/0000-0003-1169-2965}{NSF IIS-1816149, CCF-2115677, IIS-2311954, and 2421782, and  Simons Foundation MPS-AI-00010515.}

\authorrunning{P.\,M. Jacobs, F. Namjoo, and J.\,M. Phillips} %

\Copyright{Peter M. Jacobs, Foad Namjoo, and Jeff M. Phillips} %



\category{} %

\relatedversion{} %

\acknowledgements{Authors are alphabetically ordered.  Part of work completed while JMP was visiting SCaDS.AI, University of Leipzig, and MPI for Math in the Sciences.  
Thanks for Peyman Afshani for discussions about Klee's measure.  
Finally, thanks to Dagstuhl seminar 23342 on Computational Geometry of Earth System Analysis, where a conversation with Vera Schemann about distances between multi-dimensional distributions helped us realize this was the right way to generalize the KS distance -- and ultimately that these advantages were not known in the literature.  
}%

\nolinenumbers %

\EventEditors{John Q. Open and Joan R. Access}
\EventNoEds{2}
\EventLongTitle{42nd Conference on Very Important Topics (CVIT 2016)}
\EventShortTitle{CVIT 2016}
\EventAcronym{CVIT}
\EventYear{2016}
\EventDate{December 24--27, 2016}
\EventLocation{Little Whinging, United Kingdom}
\EventLogo{}
\SeriesVolume{42}
\ArticleNo{23}

\usepackage{xspace}
\usepackage{algorithm}
\usepackage{algorithmicx}
\usepackage[noend]{algpseudocode}
\usepackage{amsopn}
\providecommand{\mathbbm}[1]{\text{\usefont{U}{bbm}{m}{n}#1}}
\providecommand{\R}{\mathbb{R}}
\newcommand{\dKS}{\ensuremath{\textsf{dKS}}}
\newcommand{\mdKS}{\ensuremath{\textsf{quad\text{-}KS}}}
\newcommand{\KS}{\ensuremath{\textsf{KS}}}
\newcommand{\eps}{\varepsilon}
\newcommand{\one}{\ensuremath{\mathbbm{1}}}
\newcommand{\etal}{\emph{et.{}al.}\xspace}
\newcommand{\dir}{{\rm d}}

\newcommand{\disc}{\ensuremath{\textsf{disc}}}

\newcommand{\dgen}{\ensuremath{\texttt{\textbf{d}}}}
\newcommand{\RR}{\mathcal{R}}\newcommand{\TT}{\mathcal{T}}\newcommand{\bTT}{\bar{\mathcal{T}}}
\newcommand{\KK}{\mathcal{K}}\newcommand{\FF}{\mathcal{F}}\newcommand{\XX}{\mathcal{X}}
\newcommand{\baseline}{\dKS-\textsf{Baseline}\xspace}
\newcommand{\ourAlgo}{\dKS-\textsf{Sketch}\xspace}

\newcommand{\jeff}[1]{}\newcommand{\peter}[1]{}\newcommand{\foad}[1]{}

\begin{document}

\maketitle
\pagestyle{plain}

\begin{abstract}
We revisit extending the Kolmogorov-Smirnov distance between probability distributions to the multi-dimensional setting, and make new arguments about the proper way to approach this generalization.  
Our proposed formulation maximizes the difference over orthogonal dominating rectangular ranges ($d$-sided rectangles in $\R^d$), and is an integral probability metric.  We also prove that the distance between a distribution and a sample from the distribution converges to $0$ as the sample size grows, and bound this rate.  
Moreover we show that one can, up to this same approximation error, compute the distance efficiently in $4$ or fewer dimensions; specifically the runtime is near-linear in the size of the sample needed for that error.  With this, we derive a $\delta$-precision two-sample hypothesis test using this distance.  Finally, we show these metric and approximation properties do not hold for other popular variants. 
\end{abstract}

\section{Introduction}

This paper examines an extension of the Kolmogorov-Smirnov distance to multiple dimensions (dKS).  It also advocates for its more wide-spread use for moderate dimensions.  The dKS is non-parametric, insensitive to sampled data, and we can approximately compute it up to some specified precision in time independent of the size or complexity of the data.  We provide an asymptotic analysis of its sample complexity.  

This distance is particularly useful for multi-dimensional settings where the dimensions of the data have different units.  For example, consider comparing distributions of people measured in height and weight, or distributions of weather stations that measure temperature and pressure.  Under Euclidean distance, changes in the choice of units can alter the relative importance of the dimensions in distance comparison without changing the actual information.  This can lead to pathological problems for distances between probability measures which rely on Euclidean-based ground distance, like Maximum Mean Discrepancy or Wasserstein.  Normalization procedures (mapping to 0 mean and 1 standard deviation, or mapping to the range [0,1]) are data-dependent, may be sensitive to outliers, and can change as additional data is discovered.  
In particular, operating in units for temperature is challenging, since it does not even have a common notion of 0.  Like the univariate Kolmogorov-Smirnov distance, our notion is invariant to these units choices.

In fact, this new distance can be applied to any d-dimensional space where each coordinate has a well-defined order.  It does not require real-values or a Euclidean embedding.  For simplicity, we will focus our discussion on the real-valued case.

\subsection{The Multi-Dimensional Kolmogorov-Smirnov Distance}

Consider two probability distributions $\mu$ and $\nu$ with support in $\R^d$.  For a point $x \in \R^d$ let $x = (x_1, \ldots, x_d)$ be its $d$-coordinates.  Define $\prec$ as an operation on $p,q \in \R^d$ that returns true for $p \prec q$ if $p_j < q_j$ for all $j \in 1...d$.  Analogously $\preceq$ is true for $p \preceq q$ iff $p_j \leq q_j$ for all $j \in 1...d$.

As in Bickel~\cite{bickel1969tests}, we define the multi-dimensional Kolmogorov-Smirnov distance between two distributions $\mu, \nu$ in terms of their standard distribution functions $F_\mu(z) = \mu(\{x \in \mathbb{R}^d : x \preceq z\})$ \cite{durrett2019probability,rachev2013methods}.

\[
\dKS(\mu,\nu) = \sup_{z \in \R^d} \left| \int_{x \preceq z} \mu(x) \dir x - \int_{x \preceq z} \nu(x) \dir x \right|.  
\]

\begin{figure}[t]
\centering
\includegraphics[width=0.5\textwidth]{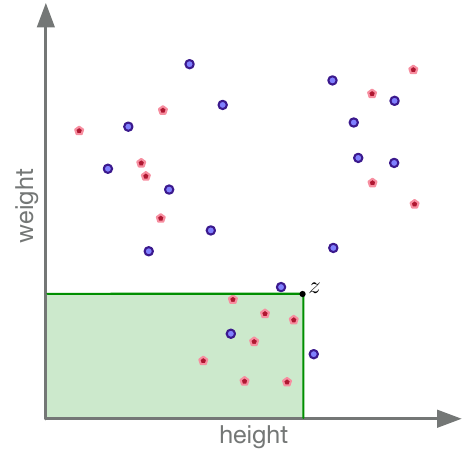}    
\caption{$\dKS(P,Q)$ uses maximizing $z \in \R^2$ to compare green region between distributions $P$ (blue $\circ$) and $Q$ (red $\star$).}
\label{fig:dKS-example}
\end{figure}

Now that we have stated the general measure theoretic form, we can restate for when the measure is defined on point sets $P,Q \subset \R^d$.  This is the natural restriction from $\mu$ to where the points $p \in P$ are Dirac deltas, each with weight $1/|P|$; and similarly, where $Q$ replaces $\nu$. We interpret point sets $P$ and $Q$ as empirical distributions with uniform weights. Then, as illustrated in Figure \ref{fig:dKS-example}, 
\[
\dKS(P,Q) = \max_{z \in \R^d} \left| \frac{|\{p \in P \mid p \preceq z\}|}{|P|} - \frac{|\{q \in Q \mid q \preceq z\}|}{|Q|} \right|.  
\]

In the $d=1$ case, both are precisely the KS distance.  
Note that, unlike the case $d=1$, the maximizing point $z$ does not need to be one of the data points from $P,Q$.  

\subsection{Properties of $\dKS$}
This paper shows the following properties
of the examined $\dKS$.  
\begin{itemize}
    \item It is an \emph{integral probability metric}, which implies that it is a  pseudometric.  Moreover, we show it satisfies the strong identity property and thus is a full metric.  
    \item It has sample complexity $n_{\eps} = O(\frac{1}{\eps^2}(d+ \log(1/\delta)))$, meaning that if one takes that many samples from any input distribution (say a continuous probability distribution), computes $\dKS$ on the sample, then with probability at least $1-\delta$, the estimated distance is off by at most $\eps$.  
    \item  We show that to compute $\dKS$ up to error $\eps$ (the best we can hope for under a sample of size $n_{\eps}$), then this can be done in time near-linear in $n_\eps$ in dimensions $d=1$ (implied by \cite{DKW56} since the 1950s), but also for $d=2,3,4$.  The respective runtimes are $O(n_\eps \log n_\eps)$, $O(n_\eps \log^3 n_\eps)$, and $O(n_\eps \log^{21} n_\eps)$ up to smaller order terms.  
    We also show that for $d > 4$, discovering an algorithm with runtime near-linear in $n_\eps$ would break a widely held conjecture in computer science, so it is unlikely to be possible.  
    \item From these algorithms we derive level-$\delta$ two-sample hypothesis tests that run in the above near-linear runtimes, times 
    $O(\log 1/\delta)$, for dimensions $d=2,3,4$.  Moreover, unlike tests based on asymptotic distributional approximations or Monte Carlo simulations, our test is precise: it guarantees an upper bound on the probability of rejection under the null hypothesis even for finite sample sizes (see Section~\ref{sec:hTesting} for details).
\end{itemize}

In another framing, this paper provides near-optimal \emph{computational-statistical runtime} (CSR) bounds for \dKS.  The goal in CSR bounds is to produce $\eps$-approximate results and $\delta$-level hypothesis tests where in addition to standard computational bounds, we allow access to probability distributions only through samples, which take $O(1)$ time each.  There is no input size and CSR bounds depend on $1/\eps$ (for $\eps$ tolerated error) and $1/\delta$ (for $1-\delta$ statistical power).  

There exists another popular higher-dimensional extension of the Kolmogorov-Smirnov distance to higher dimensions, \mdKS, which we review in the next section.  We show that this distance is not stable with respect to certain data sets; that is, adding or removing a single data point can significantly change the distance, hence there can be no associated sample complexity result as we state for $\dKS$.  
Hence, we believe our examined variant and CSR algorithm is the right approach towards this multi-dimensional generalization.

\section{Historical Context and Sample Complexity}
\label{sec:context+sample}

The Kolmogorov-Smirnov Test (the KS Test) is a long-studied non-parametric exact goodness of fit test in statistical hypothesis testing.  Given a reference distribution $\nu$ with domain $\R$ and an empirical observation $P_n$ representing $n$ iid observations from an unknown distribution $\mu$, the KS Test seeks to determine if we can be convinced that $\mu$ is distinct from $\nu$ (reject the null hypothesis that $P_n \sim \nu$).  It does so by calculating the so-called KS statistic, which is 
\[
\KS(\nu,P_n) = \max_{t \in \R} \left| \int_{x = -\infty}^t \nu(x) \dir x - \int_{x = -\infty}^t P_n(x) \dir x \right|.
\]
That is, $\KS(\nu,P_n)$ is the maximum deviation in values between the CDFs of the reference $\nu$ and empirical $P_n$ distributions.  For $P_n$ representing a probability distribution in the way specified above, and for a fixed $t$, the term $\int_{x = -\infty}^t P_n(x) \dir x$ is the fraction of observations $p \in P_n$ with value at most $t$.  

The test was proposed by Kolmogorov in 1933~\cite{Kolmogorov33}, and a table for calculating p-values published by Smirnov in 1948~\cite{smirnov1948table}.  Dvoretzky, Kiefer, and Wolfowitz in 1956~\cite{DKW56} showed remarkably that if the observations $P_n$ are indeed from $\nu$, then for any $\eps > 0$ that 
\[
\Pr[ \KS(\nu,P_n) \geq \eps ] \leq 2 \exp(-2 \eps^2 n),
\]
with constant $2$ in front of the exponential later resolved by Massart~\cite{massart1990tight}.  

Due to its simplicity, ease of calculation (algorithmically, and due to Smirnov's tables, and Dvoretzky, Kiefer, and Wolfowitz' strong asymptotic bound), and its ability to work with small $n$, this test gained popularity.  However, others have recently critiqued the test, making it unfortunately taboo~\cite{FB19,IP-KS} in some circles; these expositions focus mainly on the context of comparing to parametric distributions, since those parameters need to be estimated first, which detracts from the test's power.  
This paper does not focus on the 1-sample statistical hypothesis testing for this reason; rather we focus on its use as a distance between two distributions $\mu, \nu$ and derive a two-sample hypothesis test.  
These critiques also cite the lack of extension to more than one dimension, which this paper also addresses.

\subsection{Other Multi-Dimensional Extensions}
Others have explored high-dimensional extensions~\cite{bickel1969tests,peacock1983two,hagen2021accelerated,press1988kolmogorov,justel1997multivariate}.   
As mentioned above, $\dKS$
aligns with the formulation proposed by Bickel~\cite{bickel1969tests}.  He proposed a precise permutation two-sample test for $\dKS$; it controls type 1 error, but is computationally impractical.   It relies on an exponentially large number of $\dKS$ computations in the sample size $n$ to determine the decision threshold.  
A recent survey~\cite{stolte2024methods} reiterates that there are no practical implementations of this test.  
Moreover, via Chen and Friedman~\cite{chen2017new}, the survey states that it requires sample complexity exponential in dimension.  In contrast, this paper will show $\dKS$ can be computed efficiently, and in particular has linear sample complexity in dimension.

Separately, Peacock~\cite{peacock1983two} proposed a similar extension to the one we discuss.  There are two notable departures.  
First, Peacock was concerned with the lack of symmetry between 
$\left| \int_{x \preceq z} \mu(x) dx - \int_{x \preceq z} \nu(x) \dir x \right|$ and 
$\left| \int_{z \preceq x} \mu(x) dx - \int_{z \preceq x} \nu(x) \dir x \right|$; note the $x \preceq z$ vs. $z \preceq x$ defining the integral.  In particular, in our formulation, the choice of $z$ defines one quadrant (or in $d$-dimensions, an orthogonal $1/2^d$th) of the domain, and he argues we should consider all $2^d$ such subsets.  These can be retrieved by multiplying each different subset of the coordinates by $-1$ and then computing as we propose.  We argue that iterating through these options is not necessary to make the distance robust or well-formed.  Although, in moderate dimensions $d$, it can be handled with tolerable overhead if desired.  

Second, while the actual computation of this distance was not discussed in detail in the work by Peacock~\cite{peacock1983two}, several heuristics have become default due to available software.  As introduced by Fasano and Franceschini~\cite{fasano1987multidimensional}, focusing on $2$ dimensions, and then programmed as a \emph{Numerical Recipe} by Press and Teukolsky~\cite{press1988kolmogorov}, they proposed only evaluating quadrants where $z = p_i$ for some data point $p_i \in P$.  Now there are only $2^d n$ quadrants to consider; for constant $d$ this is linear in $n$.  We refer to this as $\mdKS$.
This has more recently been extended so it is empirically efficient in high-dimensions by Hagen \etal~\cite{hagen2021accelerated}, and shown empirically similar to the full set of quadrants Peacock proposed by Lopes \etal~\cite{lopes2007two}.   However, we will argue that this restriction to quadrants defined by the $n$ data points in $\mdKS$ is not always stable.   In particular, in Section \ref{sec:instability}, we construct settings where this restricted set of queries does not provide a stable distance under the random sample drawn $P \sim \mu$.  Hence, no sample complexity bounds can be proven. 

Note that Xiao~\cite{xiao2017fast} proposed following the original Peacock formulation in low dimensions, and provided a popular R package for $d=2$ and $d=3$.  This algorithm takes $O(n^2)$ for $d=2$ and in general $O(n^d)$ in $d$ dimensions; it saves a factor $O(n)$ off a naive approach by incrementally updating the value for adjacent quadrants considered.

Alternatively, Sadhanala \etal~\cite{sadhanala2019higher} generalize the Kolmogorov-Smirnov test to $k$-order total variation bounds, and provide efficient algorithms for small values of $k$.  We note that this is not a high-dimensional generalization, it focuses on data in $1$ dimension and allows for more general (similar to $k$-degree polynomials) test sets to measure how they differ.  Similarly, our work also makes the connection that this family of distances satisfies the criteria for an integral probability metric~\cite{muller1997integral}.

Finally, an alternative extension to higher dimensions is via a reduction to a series of 1-dimensional KS problems~\cite{justel1997multivariate,paik2025integral}.  Justel \etal\cite{justel1997multivariate} considers either a series of joint distributions or the product of all distributions.  These are computationally expensive to handle and do not directly compare the original distributions against each other.  
Moreover, it is not clear whether changing the order of processing dimensions could affect the result.  More recently, Paik \etal~\cite{paik2025integral} find the $1$-dimensional projection that maximizes discrepancy in the $k$th moment, similar to Sadhanala \etal~\cite{sadhanala2019higher}, called the Radon-Kolmogorov-Smirnov (RKS) test, and is also an IPM.  However, this RKS is not tolerant to changes in units of each coordinate -- like MMD and Wasserstein, and unlike our explored \dKS.

\subsection{Connection to Range Spaces}

The above discussion proclaimed that the Dvoretzky \etal~\cite{DKW56} sample complexity result was remarkable.  We next explain why, by showing it is part of a more general phenomenon via
the following connection to range spaces, which were mostly formalized and studied much later.  These will provide a natural way to think about the generalization of KS to higher dimensions and also justify our one-quadrant formulation.  

Consider the $d=1$ case and with a fixed value $t \in \R$.  We are interested to bound 
\begin{equation}\label{eq:KS-Chernoff}
\left| \int_{x = -\infty}^t \mu(x) \dir x - \frac{1}{|P|} \left| \{p_i \in P \mid p_i \leq t\} \right| \right|.  
\end{equation}
For each observation $p_i \in P$ we can form a binary random variable $X_i \in \{0,1\}$ that is $1$ iff $p_i \leq t$.  If we draw $P \sim \mu$, the probability that $X_i = 1$ is exactly $\int_{x = -\infty}^t \mu(x) \dir x$.  So we can apply a Chernoff bound on the average of these random variables, to determine its accuracy as measured by \eqref{eq:KS-Chernoff}: it exceeds $\eps$ with probability $2 \exp(-2 \eps^2 n)$.  This looks like the Dvoretzky \etal~\cite{DKW56} bound, yet their bound holds for all values of $t$ simultaneously.  

Leveraging a simple argument borrowed from range spaces, we can get a weaker version of the Dvoretzky \etal~\cite{DKW56} bound.  We realize that if we choose a set of $T = 2/\eps$ values $\{t_1, \ldots, t_T \}$ so that $t_i$ satisfies $\int_{x=-\infty}^{t_i} \mu(x) dx = i \cdot \eps/2$.  Then if we have at most $\eps/2$ sampling error on each of these $t_i$, any other query $t$ can be rounded up to one of these points, and incur at most $\eps/2$ rounding error in addition to the $\eps/2$ Chernoff error for $t_i$, and thus a total of $\eps$ error.  
We can extend the Chernoff bound with a union bound over these $T$ checkpoints (we do not even need to know what they are), and the probability of at most $\eps$ error on all of them is at most $2 T \exp(-2 (\eps/2)^2 n)$.  To see that this $T = 2/\eps$ term is not very significant, setting the probability of exceeding our bound to $\delta = 2 (2/\eps) \exp(- \eps^2 n / 2)$ and solving for $n$ yields $n = \frac{2}{\eps^2} \ln (4/(\delta \eps))$ where this extra factor from $T$ shows up only inside the $\log$ term.  

However, we can make another observation, that the choices of $t$ define \emph{ranges}.  
A \emph{range space} is a pair $(X,\mathcal{A})$ where $X$ is a \emph{ground set} and $\mathcal{A}$, the ranges, is a family of subsets of $X$; commonly they are defined by geometric objects such as disks, rectangles, or halfspaces.  In this case, they are one-sided intervals of the form $(-\infty, t]$.  
A celebrated theorem by Vapnik and Chervonenkis~\cite{VC71} (and refined by Li, Long, and Srinivasan~\cite{LLS01}) showed that if the range space has bounded VC-dimension $\nu$, then one can guarantee that with $n = O((1/\eps^2)(\nu + \log(1/\delta)))$ iid samples $P$ from distribution $\mu$, for each range $A \in \mathcal{A}$ that with probability at least $1-\delta$ we have
\[
\max_{A \in \mathcal{A}} \left| \int_{x \in A} \mu(x) \dir x - \frac{1}{n} |P \cap A| \right| \leq \eps.  
\]
We call a set $P$ satisfying the above condition on $(\mu,\mathcal{A})$ an \emph{$\eps$-sample} (or an \emph{$\eps$-approximation} of $(\mu,\mathcal{A})$).  
The VC-dimension for one-sided intervals is a constant, so this finally asymptotically matches the bounds of Dvoretzky \etal~\cite{DKW56}.

\subsection{Dominating Rectangular Ranges}

We next consider a family of ranges $\mathcal{R}_d$ on ground set $X \subset \R^d$, with elements $R_z \in \mathcal{R}_d$ defined by a point $z = (z_1, ..., z_d) \in \R^d$.  We call this set $\mathcal{R}_d$ the \emph{dominating rectangles} in $\R^d$.  In particular, $R_z = \{x \in \R^d \mid x \preceq z\}$.  For any point set $P \subset \R^d$, the VC-dimension $(P, \mathcal{R}_d)$ is $d$ \cite{matousek1999geometric}.  Thus we can achieve the following bound:

\begin{theorem}\label{thm:SC-Rd}
For a distribution $\mu$ on $\R^d$, then sampling $n = O((1/\eps^2)(d + \log(1/\delta)))$ points $P \sim \mu$ will, with probability at least $1-\delta$, have
\[
\dKS(\mu,P) \leq \eps.
\]
\end{theorem}

\subsection{Rectangular Range Spaces; Algorithms and Bounds}
\label{sec:rect}

While there are other commonly considered ranges (e.g., defined by halfspaces), it will be useful here to review what is known about ranges defined by rectangles.

Let $(X, \mathcal{T}_d)$ be the range space defined by axis-aligned \emph{rectangles} in $\R^d$.  As a generalization of the dominating rectangles $\mathcal{R}_d$, these are defined by two points $z,z' \in \R^d$ so $R_{z,z'} = \{ x \in X \mid z' \prec x \preceq z\}$.  The VC-dimension of this range space is $2d$ \cite{matousek1999geometric}.  Since we can set $z'$ so it is dominated by all points in $X$, when $X$ is a point set or bounded, then the ranges in $(X,\TT_d)$ are a superset of the ranges in $(X,\RR_d)$.  Hence, results which approximate $(X,\TT_d)$, which is more common in the literature, also apply directly to $(X,\RR_d)$.  

For particular range spaces, there can be more compact $\eps$-approximations (as needed for Theorem \ref{thm:SC-Rd}), than those achieved through iid sampling.  For $d > 1$, these are typically constructed via low-discrepancy colorings.  Consider a finite ground set $X$ of size $|X|=n$; if $X$ is infinite, we can first start with a sufficiently large random sample.    The objective here is to find a coloring $\chi : X \to \{-1, +1\}$ so that 
\[
\disc_\chi(X,\RR) = \max_{R \in \RR} \sum_{x \in X \cap R} \chi(x)
\]
is as small as possible.  For any range spaces $(X,\mathcal{R})$ with constant VC-dimension, then a random coloring achieves $\disc_\chi(X,\mathcal{R}) = O(\sqrt{n})$ \cite{Cha01,matousek1999geometric}; but in some cases, significantly better bounds are possible.  
As outlined in Chazelle's book~\cite{Cha01} (c.f. \cite{phillips2009small} for slight formalization and generalization), one can iteratively apply this low-discrepancy coloring to achieve $\eps$-approximation of a range space.  At a high-level, one creates a low-discrepancy coloring $\chi : X \to \{-1,+1\}$, then discards all points $X^- = \{x \in X \mid \chi(x) = -1\}$ with color $-1$.  Then repeat the process on the remaining (approximately half of the) points, until a desired size is reached.  With a few algorithmic tweaks to this process, if one can create a color $\chi$ in $O(n^\alpha \log^\beta n)$ time so that $\disc_\chi(X,\mathcal{R}) = O(\log^\gamma n)$, then in $O(n s_\eps^{\alpha-1} \log^\beta s_\eps)$ time, one can create an $\eps$-sample of size $s_\eps = O((1/\eps) \log^\gamma (1/\eps))$~\cite{phillips2009small}.  

After a long series of work, Bansal and Garg~\cite{bansal2017algorithmic} established a polynomial time $O(n^\alpha)$ algorithm for unspecified $\alpha > 1$, that can find $\chi$ such that $\disc(X, \mathcal{T}_d) = O(\log^d n)$.  More recently, Dwivedi \etal~\cite{dwivedi2019power} showed that (in an online algorithm), one could achieve $\chi$ such that $\disc_\chi(X,\mathcal{T}_d) = O(\log^{2d+1} n)$ with high probability, in $O(n \log^d n)$ time; a similar result is implied by Bansal \etal~\cite{bansal2021online}.  Combined with \cite{phillips2009small}, this implies the following result:

\begin{lemma}\label{lem:fast-eps-Td}
For range space $(X,\mathcal{T}_d)$ with point set $X \subset \R^d$ of size $n$, then in $O(n \log^d (1/\eps))$ time, one can compute an $\eps$-sample with high probability, of size $O((1/\eps) \log^{2d+1} (1/\eps))$.
\end{lemma}

\section{Computing the d-Dimensional Kolmogorov-Smirnov Distance}
\label{sec:runtime}
\vspace{-1mm}

For $1$-dimensional discrete distributions $P,Q$ (represented as point sets), computing the KS distance is fairly straight-forward.  Sort $P$ and $Q$, and then scan them in order from smallest to largest maintaining the CDFs.  Keep track of the largest difference on the points $x \in P \cup Q$ that one observes.  

However, in $d$ dimensions, there is not a single sorted order.  Yet, we can leverage an extension of an implicit observation from above: that there were at most $|P \cup Q|$ places we need to check in $\R$.  In $d$-dimensions, this scaling is not linear, but still controlled for constant $d$.  Indeed, there are at most $|P \cup Q|^d$ distinct dominating rectangular ranges, since there are $d$ coordinates in the defining vector $z$, and if the $j$th coordinate $z_j$ changes by a small enough amount that the exact same subset of the points in $P \cup Q$ have their $j$ coordinate still less than $z_j$, then it contains the same subset.  We can simply enumerate these at most $|P \cup Q|^d$ subsets, and evaluate the fraction of $P$ and of $Q$ in each to determine $\dKS(P,Q)$.

\begin{algorithm}[b]
\caption{$\text{\baseline}(P,Q)$ \;\;\;  for $P,Q  \subset \R^2$ and $|P|=|Q|=n$}

\label{alg:baseline}
\begin{algorithmic}[1]

\State Sort $P \cup Q$ by their $y$-coordinate
\State $M \gets 0$ \Comment{maximum difference found so far}

\For{(each $x$-coordinate $x_i$ in $P \cup Q$)}  \Comment{try all $x$-thresholds}
    \State $s \gets 0$ \Comment{running $P - Q$ count}
    \For{(each $j$th point in $y$-sorted order of $P \cup Q$)}
        \If{$x_j \leq x_i$}  \Comment{Otherwise point is outside the rectangle}
            \State Update $s \gets \begin{cases} s+1 & \text{ if } j \in P \\
                s-1 & \text{ if } j \in Q
            \end{cases}$ 
            \State $M \gets \max \{ M, |s| \}$
        \EndIf
    \EndFor
\EndFor

\State \textbf{return} $M /n $ \Comment{normalized KS distance}
\end{algorithmic}
\end{algorithm}

Let's say $|P| = |Q| = n$.  
Naively, since checking how many of each are in a range takes $O(dn)$ time, this process now requires $O(d^2 n^{d+1} + d n \log n)$ time.  One can reduce this to $O(d n^d + d n \log n)$ by (like in the $d=1$ case, or as in Xiao~\cite{xiao2017fast}) realizing that between two adjacent ranges, we can update their counts in $O(1)$ time. For example, fix $(d-1)$ of the coordinates (repeat this $O(n^{d-1})$ time for each possible choice) and then scan over the one remaining coordinate in sorted order using the $d=1$ algorithm.  The only difference from a direct reduction to the $1$-dimensional algorithm, is for each point, we need to verify if it satisfies the constraints from the first $d-1$ coordinates we have fixed; this can be done in $O(d)$ time.  For $d \ll n$, the runtime is dominated by still enumerating the $O(n^d)$ ranges.  
We outline the $d=2$ version of this approach in Algorithm \ref{alg:baseline} and refer to this as the \baseline algorithm within experiments in Section \ref{sec:epirical}.

\subsection{Approximate Computation}
\label{ssec:algo}

Note that with $n$ points in $P$ (and $Q$), based on the tightness of the Chernoff bound used above, the error induced on $\dKS$ from sampling is at least $\eps = 1/\sqrt{n}$. 
So computing the exact value of $\dKS$ beyond a resolution of some value $\eps$, should not be especially meaningful.  This begs the question, if we only need to produce a value $\hat v$ so  that $|\hat v - \dKS(P,Q)| \leq \eps$, then can we compute this significantly faster than $O(n^d)$?  The answer, perhaps surprisingly, is \emph{yes}!

We describe an approach laid out in work by Matheny \etal~\cite{matheny2016scalable,matheny2018computing}, designed for spatial anomaly detection.  They first argue~\cite{matheny2016scalable} we can sample two sets of points, a net $M$ of size $|M| = m = O((d/\eps) \log (d/\eps \delta))$ points, and a sample $S$ of size $|S| = s = O((1/\eps^2)(d+\log(1/\delta)))$.  Then they use $M$ to define the boundaries of the ranges to consider, and it is sufficient to evaluate those ranges on the large set $S$.  Since we want error about $\eps \approx 1/\sqrt{n}$, the second set $S$ will not be a helpful reduction for us; but the net $M$ will be.  In fact, when analyzing rectangular ranges i.e., $\mathcal{T}_d$, they further argue~\cite{matheny2018computing} that they can directly define $O((d/\eps)^d)$ ranges, instead of using the sample $M$.  Sort points along each axis, and on each axis choose $k \geq d/\eps$ evenly spaced points in the sorted order $P \cup Q$; i.e.,  for $|P \cup Q| = 2n$, select points at each index $\lfloor (i+\frac{1}{2}) \frac{2n}{k}\rfloor$ for integers $i \in \{0,1, \ldots, k-1\}$ in this order. 
The measure difference $\left| \frac{|P \cap R_{z}|}{|P|} - \frac{|Q \cap R_{z}|}{|Q|}\right|$ evaluated on any range $R_z$ versus the nearest range $R_{z'}$ from this stratified set $\mathcal{R}_M$, is $d \cdot 1/k = O(d \cdot \eps/d) = O(\eps)$; illustrated in Figure \ref{fig:range-snap}.  This can be seen as losing at most $(\eps/d)$-error along each of $d$ coordinate axes.  

\begin{figure}
\centering
\includegraphics[width=0.65\textwidth]{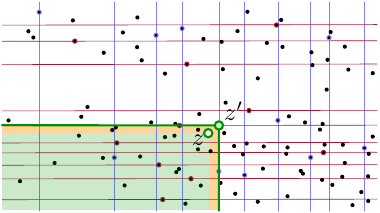}  
\caption{Adaptive grid construction with blue points evenly spaced in $x$-axis, and red ones evenly spaced in $y$-axis.  
Illustration of snapping a range defined by $z$ to one defined by $z'$.}
\label{fig:range-snap}
\end{figure}

However, we no longer want to sweep all points in one of the dimensions. This would lead to an undesired $O(n \cdot 1/\eps^{d-1})$ term in the runtime. We only want to check these $O((d/\eps)^d)$ ranges; we call these \emph{the grid}.  So we first accumulate each of $n$ points into its grid cell in $O(d)$ time, using their place in the pre-computed $d$ sorted orders along each axis; that took $O(dn \log n)$ time.  Then for each cell we can compute the total weight of $P$ and $Q$ dominated by that cell using dynamic programming in $O(1/\eps^d)$ time, and find the max difference in the same time.   This is the sketched representation of $P$ and $Q$.
Ultimately this leads to an algorithm with time $O(d n \log n + (d/\eps)^d)$.  
The pseudocode is shown in Algorithm \ref{alg:ouralgo} for $d=2$, and referred to as \ourAlgo in Section \ref{sec:epirical}.

\begin{lemma}\label{lem:fast-d}
Consider two sets of points $P,Q \subset \R^d$ with $|P \cup Q|=n$, and an error parameter $\eps > 0$.  We can compute $\hat v$ so that $|\hat v - \dKS(P,Q)| \leq \eps$ in $O(d n \log n + (d/\eps)^d)$ time.  
\end{lemma}

\begin{algorithm}[t]
\caption{$\text{\ourAlgo}(P,Q)$; \;\;\;  for $P,Q  \subset \R^2$ with $|P|=|Q|=n$, and $\eps \in (0,1)$}
\label{alg:ouralgo}
\begin{algorithmic}[1]

\State Sort $P \cup Q$ by their $x$-coordinate, and also by $y$-coordinate.  Set $k = 2/\eps$

\State Choose $k$ evenly-spaced (by rank) values $G_x = \{x_1, \ldots, x_k\}$ in $x$-order; $x_k = \max_{z \in P \cup Q} z_x$

\State Choose $k$ evenly-spaced (by rank) values $G_y = \{y_1, \ldots, y_k\}$ in $y$-order; $y_k = \max_{z \in P \cup Q} z_y$

\State Set each $G_{i,j} = 0$ for $i,j \in [0,\ldots, k]$

\For{(each $z = (z_x,z_y) \in P \cup Q$)}   \Comment{Assign points to Grid $G$}
   \State Identify $(i,j)$ so $x_{i-1} < z_x \leq x_i$ and $y_{j-1} < z_y \leq y_j$.  
   \State $G_{i,j} \gets \begin{cases} 
   G_{i,j} + 1/n &  \text{ if } z \in P \\ 
   G_{i,j} - 1/n & \text{ if } z \in Q.\end{cases}$  
\EndFor

\For{ $i=1$ \textbf{to} $k$}    \Comment{Calculate at each $G_{i,j}$ the total it dominates}
  \State $C \gets 0$
  \For{  $j=1$ \textbf{to} $k$}
      \State $C \gets C + G_{i,j}$
      \State $G_{i,j} \gets G_{i-1,j} + C$
  \EndFor
\EndFor

\State \textbf{return}  $\hat v = \max_{i,j} |G_{i,j}|$
  \Comment{Return $\varepsilon$-approximate $\dKS(P,Q)$}

\end{algorithmic}
\end{algorithm}

Now most interestingly, for $d=2$ where $\eps \approx 1/\sqrt{n}$ is the sample complexity bound, this has $O(n \log n)$ runtime, which is the same as for $d=1$.  
We can formalize the results in the computational-statistical runtime setting, as follows,  for the two-stage approach where we (1) uniformly sample $n = O((1/\eps^2)\log(1/\delta))$ points from each distribution, and then (2) use Lemma \ref{lem:fast-d} to find the approximately maximum dominating set.   

\begin{theorem}\label{thm:dKS-d}
Consider two distributions $\mu,\nu$ with support on $\R^d$ for constant $d$, and where it takes $O(1)$ time to draw a random sample.   Set two parameters $\eps, \delta > 0$.  With probability at least $1-\delta$, we can compute a value $\hat v$ so that $|\hat v - \dKS(\mu, \nu)| \leq \eps$ in time 
\[
O((1/\eps^2) \log(1/\delta)\log((1/\eps) \log(1/\delta)) + (1/\eps^d))).
\]
\end{theorem}

When restricting to $d=2$, one can save a $\log \log(1/\delta)$ by only creating a sample large enough for constant probability of success (say $1/2$), then repeating the above procedure $2\log(1/\delta)$ times, and returning the median of all results.
Ultimately, $\dKS$ in $d=2$ has the same CSR complexity as in the traditional $d=1$ case.  

\begin{corollary}
\label{cor:dKS-d=2}
Consider two distributions $\mu,\nu$ with support on $\R^2$, and two parameters $\eps, \delta > 0$.  With probability at least $1-\delta$, we can compute a value $\hat v$ so that $|\hat v - \dKS(\mu, \nu)| \leq \eps$ in $O((1/\eps^2)\log(1/\eps) \log(1/\delta))$ time.  
\end{corollary}

\subsection{Near-Linear Time in 3 and 4 Dimensions}
\label{sec:near-lin34}

Invoking Theorem \ref{thm:dKS-d} with $d=3$ requires $\Theta(1/\eps^3)$ time which is larger than the sample complexity of $n=O((1/\eps^2)\log(1/\delta))$.  Can we reduce the runtime to $\eps$-approximate $\dKS(\mu,\nu)$ back to be near-linear in the sample complexity?  
Through a more complicated approach, we show here that we indeed can, even for $d=4$.

First we sample $X \sim \mu$ and $Y \sim \nu$, each of size $n = O((1/\eps^2)\log(1/\delta))$.  
Then we will invoke algorithms implied by Lemma \ref{lem:fast-eps-Td} with constant probability of success (specifically with $\delta=1/2$), and later will obtain full probabilistic bounds using the median trick; see Algorithm \ref{alg:dks-Klee}.  For brevity, we omit probabilistic statements until then.  In particular, we
create a set $P$ and $Q$ (from $X$ and $Y$, respectively) of size $s = O((1/\eps) \log^{7}(1/\eps))$ so that for all $R \in \mathcal{T}_d$ that $\left| \frac{|X \cap R|}{|X|} - \frac{|P \cap R|}{|P|} \right| \leq \eps$; and similar for $Y,Q$.  This takes $O((1/\eps^2) \log^3(1/\eps))$ time.  
For $d=4$ this has the same guarantees using $s_4 = O((1/\eps) \log^9(1/\eps))$ and in time $O((1/\eps^2) \log^4(1/\eps))$.  
We now operate just on $P,Q$ and approximate their distance within $\eps$ error.  This results in a total of $3 \eps$ error, but can be adjusted by setting $\eps' = \eps/3$ and achieving overall error $\eps'$ at each step with the same asymptotics as the $\eps$ error we analyze.

Next, we convert this problem of computing $\dKS(P,Q)$ into another computational geometry problem about a set of $s$ rectangles in $\R^d$, among which we want to find a point $p \in \R^d$ of maximum depth: the point $p$ that is contained in most rectangles.  This problem has a long history and is deeply related to the so-called Klee's measure problem (c.f. \cite{chan2013klee}).  Notably, it can be solved in roughly $s^{d/2}$ time, which will provide asymptotically-near-optimal approximate CSR algorithms for $d=3,4$.  The specific result we need involves the $s$ rectangles in $\R^d$, each with a weight, and the goal is to find a point $z \in \R^d$ that maximizes the sum of the weights of the rectangles that contain it.  Chan~\cite{chan2013klee} states the results with a parameter $b \geq 1$, and shows that this can be solved in $O(s^{d/2} b^{5-d/2} \log b + b^{O(b)})$ time for any constant $d \geq 3$.  Setting $b = O(\log s / \log \log s)$ yields the following result.  

\begin{lemma}[Chan~\cite{chan2013klee}]
\label{lem:max-rect}
Given a set of $s$ rectangles in $\R^d$ with weights, in $O(s^{d/2} \log^{5-d/2} s \cdot (\log \log s)^{O(1)})$ time we can find a point $z$ which maximizes the sum of those weights, and also return that maximal weight.  
\end{lemma}

For $d=3$ the runtime is $O(s^{1.5} \log^{3.5} s (\log \log s)^{O(1)})$, and for $d=4$ the runtime is $O(s^2 \log^3 s (\log \log s)^{O(1)})$.

\begin{figure}
\centering
    \includegraphics[width=0.5\textwidth]{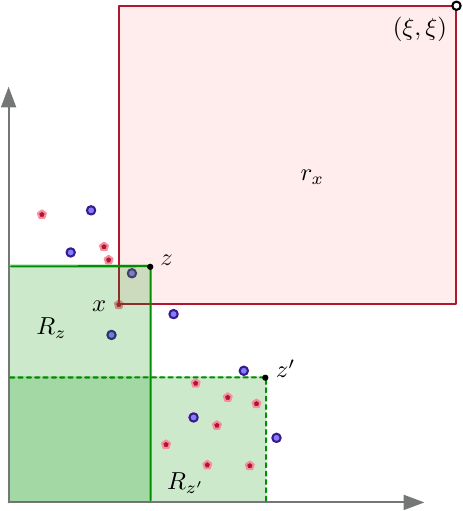}
    \caption{Mapping to Klee's problem in $d=2$: $x \in R_z$ (green) iff $z \in r_x$ (red rectangle)}
    \label{fig:Klee-dual}
\end{figure}

The reduction is now fairly straightforward.  
Let $\xi = 2\max_{j \in [1,d]} \max_{x \in P \cup Q} x[j]$, where $x[j]$ is the $j$th coordinate of $x$.  We map each $x \in P \cup Q$ to a rectangle $r_x = \{z \in \R^d \mid x[j] \leq z[j] \leq \xi \textrm{ s.t. } j \in [1,d]\}$.  Now $R_z$ includes $x$ (in the primal) if and only if $r_x$ contains $z$ (in the dual); illustrated for $d=2$ in Figure \ref{fig:Klee-dual}.  
For each $x \in P$ we set the weight to $+1/|P|$, and for $x \in Q$ at $-1/|Q|$; then for a point $z \in \R^d$, the total weight of the rectangles it contains is the difference in the fraction of points $R_z$ contains from $P$ versus $Q$.  That is, with $P,Q$ found via Lemma \ref{lem:fast-eps-Td}, then the maximum value among all $z$ is an $\eps$-approximation of $\dKS(\mu, \nu)$.  

Finally, to understand the runtime of approximating $\dKS$ within $\eps$, we can invoke Lemma \ref{lem:max-rect} with $s = |P \cup Q| = O((1/\eps) \log^{2d+1} (1/\eps))$.  Recall that we set the initial sample of $X,Y$ to size $n=O(d/\eps^2)$, so it succeeds with constant probability in achieving our $\eps$-approximation property.  Then we reduced the size $s = |P \cup Q|$ using Lemma \ref{lem:fast-eps-Td}.  We repeat the entire endeavor (sample $n$ points, compress to $s$, apply Lemma \ref{lem:max-rect} in the dual) $2\log 1/\delta$ times, and return the median answer; which has error at most $\eps$ with probability at least $1-\delta$.  This is sketched in Algorithm \ref{alg:dks-Klee}.  

\begin{theorem}\label{thm:d}
Consider two distributions $\mu,\nu$ with support on $\R^d$ (for constant $d$), and two parameters $\eps, \delta \in (0,1]$.  With probability at least $1-\delta$, we can compute a value $\hat v$ such that $|\hat v - \dKS(\mu, \nu)| \leq \eps$ in 
$O(((1/\eps^2)\log^d (1/\eps) + (1/\eps)^{d/2}\log^{d^2+5}(1/\eps)(\log \log (1/\eps))^{O(1)})\log(1/\delta))$ time.  
\end{theorem}

\begin{algorithm}[t]
\caption{approx-$\dKS(\mu,\nu)$; \;\;\;  for $\mu,\nu$ defined on $\R^d$ with $d=3,4$, and $\eps,\delta \in (0,1)$}
\label{alg:dks-Klee}
\begin{algorithmic}[1]

\For{$j = 1$ to $2\log(1/\delta)$}

    \State Sample $X \sim \mu$ and $Y \sim \nu$ as i.i.d. samples of size $n = O(1/\eps^2)$ 

    \State 
    Reduce $X,Y$ to $P,Q$ using Lemma~\ref{lem:fast-eps-Td};  
    each of size $s = O((1/\eps) \mathrm{poly}\log(1/\eps))$ 

    \State Convert each point $x \in P \cup Q$ to a rectangle $r_x$ in dual %

    \State Run Chan's Algo (Lemma \ref{lem:max-rect}), to find a max weight point $z \in \R^d$ 

    \State Estimate $\hat v_j = \left | \frac{|P \cap R_z|}{s} - \frac{|Q \cap R_z|}{s} \right|$ for rectangle $R_z$ in primal. %

\EndFor

\State \textbf{return} $\hat v = \textsf{median}(\{\hat v_1, \hat v_2, \ldots \})$
\Comment{Return $(\varepsilon,\delta)$-approx of $\dKS(\mu,\nu)$}

\end{algorithmic}
\end{algorithm}

The runtime comes from $O((1/\eps^2) \log^d(1/\eps))$ to invoke Lemma \ref{lem:fast-eps-Td}.  
And then another
\begin{equation*}
\begin{split}
O(s^{d/2} \log^{5-d/2} s\,(\log\log s)^{O(1)})
&= O\!\Big((1/\eps)^{d/2}\, \log^{(2d+1)d/2}(1/\eps) \\
&\qquad\quad \log^{5-d/2}(1/\eps)\,(\log\log(1/\eps))^{O(1)}\Big)
\end{split}
\end{equation*}
to invoke Lemma \ref{lem:max-rect}. We can state the following corollaries for $d=3$ and $d=4$, which achieve a computational-statistical runtime near-linear in the sample complexity $O((1/\eps^2)\log(1/\delta))$.  

\begin{corollary}
\label{cor:dKS-d=3}
Consider two distributions $\mu,\nu$ with support on $\R^3$, and two parameters $\eps, \delta \in (0,1]$.  With probability at least $1-\delta$, we can compute a value $\hat v$ so that $|\hat v - \dKS(\mu, \nu)| \leq \eps$ in 
$O((1/\eps^2)\log^3 (1/\eps) \log(1/\delta))$ time.  
\end{corollary}

\begin{corollary}
\label{cor:dKS-d=4}
Consider two distributions $\mu,\nu$ with support on $\R^4$, and two parameters $\eps, \delta \in (0,1]$.  With probability at least $1-\delta$, we can compute a value $\hat v$ so that $|\hat v - \dKS(\mu, \nu)| \leq \eps$ in 
$O((1/\eps^2)\log^{21} (1/\eps)(\log^{O(1)}\log(1/\eps)) \log(1/\delta))$ time.  
\end{corollary}

\subsection{Hardness of Improving Runtime}
There is evidence that this runtime cannot be significantly improved by reducing to a well-studied problem.  
Specifically, \textsc{Max-Weight $k$-Clique} is a problem on $k$-partite graphs with $k$ disjoint vertex sets of size $n$ where edges must have endpoints in disjoint sets. The edges are assigned weights, and the goal is to find one vertex in each set (a $k$-clique) so summing over the induced edges has maximum weight.  Despite much effort, there are no algorithms for this problem which improve polynomially on an $O(n^k)$ runtime.  The fine-grained complexity project has built small-polynomial-time reductions between numerous problems, so a polynomial improvement in one would imply a polynomial improvement for many others.  

Backurs \etal~\cite{backurs2016tight} show that if one can achieve a $O(n^{d-\alpha})$ runtime for any $\alpha > 0$ to find the rectangle $R \in \mathcal{T}_d$ which maximizes the difference in total weight included between two size-$n$ weighted point sets $P,Q \subset \R^d$, then one can solve the \textsc{Max-Weight $k$-Clique} problem in $O(n^{k-\alpha})$ time for $k = \lceil d^2/\alpha \rceil$ -- which is assumed hard.  
Matheny and Phillips~\cite{matheny2018computing} observe that if one can solve the $\eps$-approximate version in $O(n + 1/\eps^{d-\alpha})$ time (for any $\alpha > 0$) then setting $\eps$ sufficiently small (they set $\eps = 1/(4n)$), one can also solve the exact version.  So this implies conditional hardness for the challenge of finding the maximum difference rectangle from $\mathcal{T}_d$ in faster than $O(n + 1/\eps^d)$ considering any $\eps > 0$.  

To apply this to the dKS problem, we need to address two issues.  First, we can only set $\eps \approx 1/\sqrt{n}$, and second convert from $\mathcal{T}_d$ (with $2d$ sides in $\R^d$) to $\mathcal{R}_d$ (with $d$ sides in $\R^d$).

\begin{figure}
    \centering
    \includegraphics[width=0.5\textwidth]{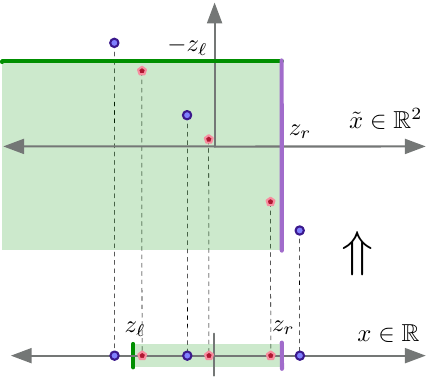}
    \caption{Lifting from intervals (1-d rectangles) $\TT$ in $\R$ to dominating rectangles $\RR$ in $\R^2$.}  %
    \label{fig:lift-2d}
\end{figure}

To address the second issue, we can map the full rectangle $\mathcal{T}_d$ problem in $\R^d$ into a dominating rectangle $\mathcal{R}_{2d}$ problem in $\R^{2d}$.  For each point $x = (x_1, x_2, \ldots x_d) \in \R^d$ map to a point $\tilde x \in \R^{2d}$ where $\tilde x_i = x_i$ for $i \in 1 \ldots d$ and $\tilde x_i = -x_{i-d}$ for $i \in d+1 \ldots 2d$.  Then $x \in R_{z,z'}$ iff $\tilde x \in R_{\tilde z}$ where $\tilde z \in \R^{2d}$ is the concatenation of $z$ and $-z'$.  One can see this in $d=1$ (in Figure \ref{fig:lift-2d}), where the mapped points $\tilde x$ in $\R^2$ all lie on the $x_1 = -x_2$ line through the origin, and $2$-sided dominating rectangles $R_{(z_r, -z_\ell)}$ in $\R^2$ can only form the same subsets as intervals $[z_\ell, z_r]$ in $d=1$.  Repeating this argument for each of the original $d$ coordinates can complete the reduction.  So following Backurs \etal~\cite{backurs2016tight}, if we can solve the exact $\dKS(P,Q)$ on point sets $P,Q \in \R^{2d}$ each of size $n$ in time $O(n^{d-\alpha})$, then we can solve the \textsc{Max-Weight-$k$-Clique} problem in $O(n^{k-\alpha})$ for $k = \lceil d^2/(4\alpha) \rceil$.

To address the first issue, we can consider any exact problem on $m$ points $P,Q$ in $\R^d$ where we wish to find the rectangle $R \in \TT_d$ that maximizes the difference between $P$ and $Q$.  We convert this into a problem on $n = c m^2$ points, for some constant $c$, where we seek to solve with $\eps = 1/(4\sqrt{cn}) = 1/4m$ error.  The conversion simply replaces each of the $m$ points in the original problem with $cm$ copies of the same point to form the size-$n$ problem, so any rectangle $R \in \TT_d$ that contains one such point must contain all of them in the set of $cm$ copies.  Now to achieve $\eps$-error on the size-$n$ problem, we must solve the size-$m$ problem exactly.  This overall construction leads to the following theorem.

\begin{theorem}\label{thm:hardness}
For two point sets $P,Q \in \R^d$ of size $n$ if we can find a $\hat v$ so that $|\hat v - \dKS(P,Q)| \leq \eps$ for $\eps = O(1/\sqrt{n})$ in time $O(n^{d/4-\alpha})$ for $\alpha > 0$, then we can solve \textsc{Max-Weight-$k$-Clique} on $m = \Omega(\sqrt{n})$ points in $O(m^{k-\alpha})$ time for $k = \lceil d^2/(16\alpha) \rceil$.  
\end{theorem}

Hence, if the goal is to solve for an $\eps$-approximation $\hat v$ to the $\dKS(P,Q)$ problem in $O(n ~\mathrm{polylog}(n))$ time in $\R^d$ for $\eps = O(1/\sqrt{n})$, it is unlikely it can be accomplished in time near-linear in $n$ when $d > 4$, unless there is a major breakthrough in the \textsc{Max-Weight-$k$-Clique} problem.  In particular, conditioned on this assumption, up to sublinear factors, we have established the computational-statistical runtime of $\tilde\Theta(1/\eps^2 + 1/\eps^{d/4})$ to $\eps$-approximate $\dKS$ in $\R^d$.

\section{Instability of Alternative Multi-Dimensional Kolmogorov-Smirnov Distance Formulations}
\label{sec:instability}

Here we discuss how common heuristic extensions of the KS distance to higher dimensions~\cite{peacock1983two,fasano1987multidimensional,press1988kolmogorov,hagen2021accelerated,lopes2007two,xiao2017fast} -- in particular $\mdKS$, which only considers subsets where the range corner must be a data point -- results in an estimate unstable with respect to the sampled data.  In contrast, for $\dKS$ we prove the sample complexity formalized in Theorem \ref{thm:SC-Rd}.  An implication of the instability demonstrated in the constructions that follow is that such a general result is not possible for $\mdKS$.  

Consider two point sets $P,Q \subset \R^2$ that lie very close to the line $x=y$ (with some noise off that line); it is rare that the second or fourth quadrant defined with a data point (from $P \cup Q$) in the corner contains more than 1 point (or say a constant number of points).  
Thus, the points effectively lie on a 1-dimensional subspace, and only the first and third quadrants matter.  Then in the middle of this sequence, consider a range that contains an $\alpha$-fraction of both $P$ and of $Q$ (think of $\alpha = 1/2$).  However, this interval is structured so we first see an $\alpha/2$-fraction of points from $Q$, then an $\alpha$-fraction from $P$, and then an $\alpha/2$-fraction from $Q$ again; see Figure \ref{fig:hard-cases}(left).  For any other interval, the distributions are basically balanced.  The largest quadrant-defined subset among these points has a KS distance of at most $\alpha/2$: containing just one $
\alpha/2$-fraction section of $Q$ (e.g., the orange range in Figure \ref{fig:hard-cases}(left)) or that and the $\alpha$-fraction section of $P$.  
However, this is unstable to the sampling of points: consider a single additional point $p^\dagger$ which is off of this $x=y$ line so that its second quadrant contains exactly the $\alpha$-fraction section of $P$ (e.g., the green range in Figure \ref{fig:hard-cases}(left)).  Now the KS distance is $\alpha$, a factor of $2$ increase from before.  

But should it matter if $p^\dagger$ is in the point set or not?  Consider a distribution where the probability of sampling a point at $p^\dagger$ is $1/(2n)$.  Then if we sample $n$ points, with probability roughly $1/2$ the point at $p^\dagger$ is included (and $\KS(P,Q) \approx \alpha$) and probability roughly $1/2$ no such point $p^\dagger$ is included (and $\KS(P,Q) \approx \alpha/2$).  

If we only consider a single quadrant direction, the contrast is even higher.  By symmetry consider the second quadrant, and then either no quadrant contains more than a constant number of points (so $\KS(P,Q) \approx 1/n$) or if $p^\dagger$ exists the offending $\alpha$-fraction is in a range (so $\KS(P,Q) = \alpha$).

\begin{figure}
\hfill
\includegraphics[width=0.4\textwidth]{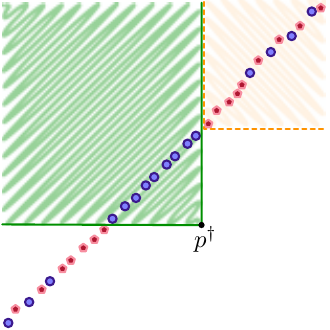}   
\hfill
\includegraphics[width=0.4\textwidth]{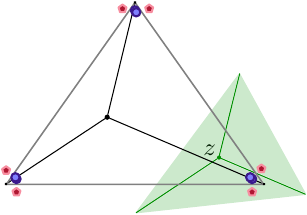}    
\hfill
\phantom{.}
\caption{Hard example in $d=2$ (left) and $d=3$ (right) between distributions blue $\circ$ and red $\star$.}
\label{fig:hard-cases}
\end{figure}

In higher-dimensions, even with evaluating all quadrants, this gap between including $p^\dagger$ or not becomes more extreme.  Now consider the base distribution containing points very near to the $(d-1)$-dimensional subspace where $\sum_{j=1}^d x_j = 1$; if we further restrict so that each $x_j \in [0,1]$, this is sometimes called the \emph{$(d-1)$-simplex}.  With base point $p^\dagger = 0$ (the origin), then its positive orthant contains exactly this $(d-1)$-simplex (so each $x_j \in [0,1]$) and nothing else.  Now we arrange our $\alpha$-fraction of $P$ so that it lies in this simplex, and specifically a $\alpha/d$-fraction near each vertex.  For the mass near vertex $e_j = (0,0,\ldots, 0,1,0,\ldots 0)$ where the $1$ is in the $j$th spot, the corresponding mass of $P$, for some small value $\eta$, lies at $(\frac{\eta}{d-1}, \frac{\eta}{d-1}, \ldots, \frac{\eta}{d-1}, 1-\eta, \frac{\eta}{d-1}, \ldots, \frac{\eta}{d-1})$.  
A corresponding $\alpha$-fraction of $Q$ is placed also near each vertex, but just outside the simplex.  For each corner, the mass of $Q$ is not placed in just one location, but is split among $d-1$ different directions each offset from the corner of the simplex along a different coordinate, and each containing a $\alpha/(d(d-1))$ fraction of the mass.  For vertex $e_j$, the $\ell$th offset part of $Q$ is at $(0, \ldots, 0, -\eta, 0, \ldots, 1+\eta, 0, \ldots, 0)$, where $1+\eta$ is in the $j$ coordinate, and the $-\eta$ is in the $\ell$th coordinate, and $j \neq \ell$.  This is illustrated in Figure \ref{fig:hard-cases}(right) for $d=3$.  
Now any axis-aligned query with a corner in the restricted plane, that contains the $(d-1)$-simplex and all data described so far, must contain at least $(d-2)$ of the $(d-1)$ corresponding offsets of $Q$.  However, centering at $p^\dagger$ (= the origin) can contain all of the $\alpha$-fraction of $P$, and none of the $\alpha$-fraction of $Q$.  The remainder $(1-\alpha)$ mass of $P$ and $Q$ can be paired with each other.  
Overall, this implies that the difference in what can be captured from a query with $p^\dagger$ and one restricted to the $\sum_j x_j=1$ set is $\alpha (d-1)/d$ and approaches $\alpha$ as $d \to \infty$.

While these pathological cases are presented through very precisely constructed cases, they illustrate two things.  First, $\mdKS$ cannot have generic sample complexity results as does $\dKS$ in Theorem \ref{thm:dKS-d}.  
Second, this removes blanket statements that $\mdKS$ can be applied regardless of distributions, since for these specific distributions it cannot provide consistent answers.

\section{$\dKS$ as an Integral Probability Metric}
\label{sec:IPM+RS}

In this section we make an argument that the proposed dKS formulation using only one-quadrant ranges ($\RR_d$, not all $2^d$ options) is perfectly reasonable.  As an implication, due to its efficiency and invariance to units along its axes, it should be a preferred choice in many contexts.   
To make this ``reasonableness'' claim, we show that $\dKS$ is from a well-studied distance family (whereas other variants are not).

An \emph{integral probability metric}~\cite{muller1997integral} (or IPM) is a family of distances between two distributions $\mu, \nu$ defined on a metric space $\XX$.  It is specified then by a family of real-valued functions $\FF$ on $\XX$.  Then the distance $\dgen_{\FF}(\mu, \nu)$ is defined
\[
\dgen_\FF(\mu,\nu) = \sup_{f \in \FF} |\mathsf{E}_{x \sim \mu}f(x) - \mathsf{E}_{x \sim \nu} f(x)|,
\]
where $\mathsf{E}_{x \sim \mu} f(x)$ is the expectation of the value $f(x)$ where $x$ is drawn from $\mu$.  

M\"uller showed that all IPMs are pseudometrics~\cite{muller1997integral}, and may be metrics depending on the choice of family $\FF$.  This means that $\dgen_\FF$ satisfies 
(a) symmetry $\dgen_\FF(\mu, \nu) = \dgen_\FF(\nu,\mu)$,
(b) triangle inequality $\dgen_\FF(\mu, \xi) + \dgen_\FF(\xi, \nu) \geq \dgen_\FF(\mu,\nu)$, and
(c) the weak version of identity that if $\mu = \nu$ then $\dgen_\FF(\mu,\nu) = 0$.  
IPMs can be metrics and satisfy the strong form of identity that also $\dgen_\FF(\mu,\nu) = 0$ implies that $\mu = \nu$.  This holds for other IPMs, such as total variation, where $\FF$ is the space of indicator functions on the Borel sets.  The kernel distance $\dgen_\KK$
(also known as maximum mean discrepancy~\cite{gretton2006kernel}) is defined with $\FF = \KK = \{\int_{q \in \R^d} K(\cdot, q) d \mu(q)\}$ over metric space $\XX = \R^d$.  When the family of kernels is ``characteristic''~\cite{sriperumbudur2011universality}, which includes the Gaussian and Laplace kernels, then $\dgen_\KK$ is also a metric.  
These functions could also more generally allow for a variety of other families of classifiers used in machine learning~\cite{lopez2017revisiting,kim2016classification}.

An important class of IPMs is when we define $\FF$ as an indicator function induced by a combinatorial range space $(\R^d, \RR)$.  Then $f(\cdot) = \one_R(\cdot)$ for range $R \in \RR$ and $\one_S$ is the indicator function which is $1$ iff the argument is in the set $S$, and $0$ otherwise. In this case for metric space $\XX = \R$, then the original KS distance is defined with $\dKS = \dgen_{\RR_1}$, where $I_z \in \RR_1$ is a one-sided interval $(-\infty, z]$.  Then it is known that $\dgen_{\RR_1} = \dKS$ is a metric.
However, applying ranges $\RR_1$ to only the first coordinate of $\R^d$ for $d>1$ does not lead to a metric. It is immediate that $\dKS$ is an IPM defined by the indicator functions on the sets in $\mathcal{R}_d$.

For $d \geq 1$, $\dKS$ is indeed a metric over the probability measures defined on $\mathcal{B}(\R^d)$, the Borel sets on $\mathbb{R}^{d}$; while this is well-known in probability theory, it is important and thus is formally stated below in Theorem \ref{thm:dKS-metric}. Proving this for distributions $\mu, \nu$ as discrete sets $P,Q$ (with uniform weight) has a relatively short and simple proof, by showing if $\dKS(P,Q) = 0$ then $P=Q$.  This holds by contrapositive: if $P \neq Q$, then let $z$ be a lexicographically minimal point where they are different; then $|P \cap R_z| \neq |Q \cap R_z|$ and so $\dKS(P,Q) \neq 0$.  
The proof for the probability measures $\mu,\nu$ defined on Borel sets is slightly more technical, and although it is known in probability theory, for completeness' sake it is included in Appendix \ref{app:metric}.

\begin{theorem}
    \label{thm:dKS-metric}
    On the collection of probability measures defined on $(\R^d,\mathcal{B}(\R^d))$, $\dgen_{\RR_d} = \dKS$ is a metric.
 \end{theorem}

Note that the heuristic $\mdKS$ approaches, which only consider quadrants anchored by sample points, do not define $\FF$ independent of the discrete measures being considered, and hence do not automatically inherit the pseudometric properties like the triangle inequality.

\section{Application to Two-Sample Hypothesis Testing}

\label{sec:hTesting}

In the two sample hypothesis testing problem, the practitioner has collected $ Y_1,Y_2,\dots,Y_m \overset{iid}{\sim} \mu$ and $X_1,X_2,\dots,X_n \overset{iid}{\sim} \nu$ and the goal is to determine from these samples whether $\mu = \nu$. In classical terminology, the \textit{null hypothesis} $H_0$ is that $\mu = \nu$ and the \textit{alternative hypothesis} is that $\mu \neq \nu$. In the current work, $\mu$ and $\nu$ are supported on $\mathbb{R}^{d}$ for some $d \in \{1,2,3,4\}$ and we assume for simplicity that the sizes of the samples are the same (i.e., $m=n$). Letting $D_n$ be any function of the two samples with discriminatory power, for any $0 < \delta < 1$ a level-$\delta$ test is constructed by determining a threshold $\eps(n,\delta)$ such that
\[
\mathbb{P}_{H_0} \Big(D_n(Y_1,Y_2,\dots,Y_n,X_1,X_2,\dots,X_n) \geq \eps(n,\delta) \Big) \leq \delta.
\]
A test that then rejects $H_0$ when $D_n$ exceeds $\eps(n,\delta)$ is called a \textit{level-$\delta$ hypothesis test} and $D_n$ is the \textit{test statistic}. Such a test is finite sample valid in that the probability of rejection given the null hypothesis is true is exactly upper bounded by $\delta$. An important consideration when designing a test is the worst-case runtime to conduct a level-$\delta$ test (for $0<\delta<1$) with $n$ samples, a function of $n$ and $\delta$, depending on the runtimes to compute the test statistic and the threshold $\eps(n,\delta)$.

\subsection{Using $\dKS$ for Two-Sample Testing}

Aided by Theorem \ref{thm:SC-Rd}, Corollaries \ref{cor:dKS-d=2}, \ref{cor:dKS-d=3}, and \ref{cor:dKS-d=4} hold with $n = \Omega(\frac{1}{\eps^2} \log(1/\delta))$, respectively for $d \in \{2,3,4\}$; which implies there is some constant $C_{d}$ such that this holds for $n \geq C_{d} \frac{1}{\eps^2} \ln(1/\delta)$. Using this and that $\dKS$ is a metric, we have that for any $n$, $0<\delta<1$, the approximate $\dKS$ distance $\hat{v}$ (from Corollaries \ref{cor:dKS-d=2}, \ref{cor:dKS-d=3}, and \ref{cor:dKS-d=4} respectively), can be used as a test statistic with decision threshold 
\[
\eps(n,\delta) := \sqrt{\frac{C_{d} \ln(1/\delta)}{n}}.
\]

Plugging $\eps \equiv \eps(n,\delta)$ yields a level $0<\delta<1$ test using $\hat{v}$ with runtime 
$O(n \log n \log(1/\delta))$ when $d=2$, 
$O(n \log^3 n \log(1/\delta))$ when $d=3$ and 
$O(n \log^{21} n \log^{O(1)} \log n \log(1/\delta))$ when $d=4$.  
A comparison of the runtime of two-sample hypothesis testing procedures using a version of multidimensional KS ($\mdKS$) as a function of $(n,\delta)$ is provided in Table \ref{tab:testingComp}. Our approach is the first to provide a stable version (see Section \ref{sec:instability}) of multidimensional $\KS$ with near-linear runtime in dimensions $d=2,3,4$.

\begin{table}[t] %
\resizebox{\columnwidth}{!}{%
\begin{tabular}{cccccc}
\hline
\textbf{Citation} & \textbf{Test Statistic} & \textbf{Stable} & \textbf{Precise Test} &  $(n,\delta)$ \textbf{Test Runtime} 
\\ \hline
\cite{peacock1983two,xiao2017fast} & Peacock & Yes & No & $O(n^{d})$    
\\ 
\cite{fasano1987multidimensional,lopes2007two,hagen2021accelerated} & $\mdKS$ & No & No & $O(n \log n)$     
\\ 
this paper & $\dKS$ & Yes & Yes 
& $O(n\log n \log(1/\delta))$ & $d=2$ 
\\ 
this paper & $\dKS$ & Yes & Yes 
& $O(n \log^3 n \log(1/\delta))$ & $d=3$ 
\\ 
this paper &$\dKS$ & Yes & Yes 
& $O(n \log^{21 + \zeta} n \log(1/\delta))$ & $d=4$ 
\\ \hline
\end{tabular}}
\caption{Comparison of Kolmogorov-Smirnov two-sample testing procedures by test statistic and runtime for a constant dimension $d > 1$ since \cite{bickel1969tests} initiated the study.   The last bound holds for any $\zeta >0$ to hide $\mathsf{poly}(\log \log n)$ terms. By \textit{Precise Test}, we mean the probability of rejection under the null hypothesis is exactly upper bounded by $\delta$ without asymptotic or simulation-based approximation. }
\label{tab:testingComp}
\end{table}

\subsection{Precision in Testing}
To apply a two-sample hypothesis test at a $\delta$ level, one needs to calculate a value $\eps(n,\delta)$ for the decision threshold.  In the one-dimensional $\KS$, famously Smirnov~\cite{smirnov1948table} provided tables, and Dvoretzky, Kiefer, and Wolfowitz~\cite{DKW56} provided general bounds.  

However, in prior work in the multi-dimensional case, these values have only been estimated.  Specifically, 
Bickel~\cite{bickel1969tests} showed $\dKS$ was well-defined for this task and provided a permutation test with guaranteed type I error control in finite samples, but the test requires an exponentially large number of $\dKS$ computations (relative to the sample size $n$), rendering this exact approach unusable even for modestly large sample sizes.  
Peacock \cite{peacock1983two} makes an educated guess of the asymptotic distribution of the test statistic after Monte Carlo simulations. For $\mdKS$, Fasano \etal \cite{fasano1987multidimensional} tabulates the critical value tables (as a function of $n,\delta$) via Monte Carlo simulation for their test statistic. Press \cite{press1988kolmogorov} follows \cite{fasano1987multidimensional} with an educated guess of the distribution of the test statistic for $\mdKS$.  Hagen \textit{et.al.} \cite{hagen2021accelerated} arguably provide the most faithful finite sample approximation to the $\mdKS$ statistic, providing an analytical approximation for the cumulative distribution function of the test statistic which only differs from the precise distribution in that orthant volumes are approximated using the samples.

To allow finite sample validity (specifically an exact guarantee on the upper bound of the rejection probability when the null hypothesis is true) in tests with $\dKS$, we need to provide an upper bound for the constant $C_d$.  
This involves two aspects, the sample complexity bounds with respect to the range space $(\R^d, \RR_d)$, and then the approximation from only searching regions in a sparse grid over the sample.  Section \ref{ssec:algo} explained how to achieve $\eps/2$ error in $T$ ranges, applying a Chernoff and Union bound, that $n \geq \frac{2}{\eps^2} \ln (\frac{1}{\delta T})$ samples are sufficient with precision level $\delta$.  The sparse grid we consider in Section \ref{sec:runtime} induces $d \cdot (d/\eps')^d = T$ such ranges, so any range $R_z \in \RR_d$ is within $\eps'$ of some induced range.  Setting $\eps'=\eps/2$ we get $T = d^{d+1}2^d/\eps^d$ and $n \geq \frac{2}{\eps^2} \ln(\frac{1}{\delta}\frac{d^{d+1}2^d}{\eps^d})$ is sufficient.  Setting $\eps = \sqrt{\frac{C_d \ln(1/\delta)}{n}} = \sqrt{\frac{\ln \frac{1}{\delta T}}{n}}$ and using the bound $T = d^{d+1} 2^d /\eps^d$ to 
solve for $C_d \leq 2 \ln(d^{d+1}2^d/\eps^d) = 2d \ln(d 2 d^{1/d}/\eps) \leq 2d \ln(4d/\eps)$ does not provide an absolute constant since it still depends on $\eps$.  However, this bound on $C_d$ is useful for any non-microscopic values of $\eps$.  

To derive a decision threshold using $C_d = \alpha_d \ln(\beta_d / \eps)$ with some constants $\alpha_d$ and $\beta_d$ for a fixed small $d$, we then need to solve for $\eps$ so its only on the left-hand-side.  Now as long as $n$ is sufficiently large, so $\ln(\beta_d^2 n/\alpha_d)\ln(1/\delta) < \sqrt{n/\alpha_d}$, then 
\begin{align*}
\eps (n,\delta) 
&\geq 
\sqrt{\frac{\alpha_d \ln(\beta_d/\eps)}{n} \ln(1/\delta)} 
\\ & \geq 
\sqrt{\frac{\alpha_d \ln \left(\frac{\beta_d \sqrt{n} \ln(\beta_d/\eps)}{\sqrt{\alpha_d}}\ln(\frac{1}{\delta}) \right)}{n}\ln(1/\delta)} 
\\ & \geq 
\sqrt{\frac{\alpha_d \ln (\beta_d n /\alpha_d)}{n} \ln(1/\delta)}.
\end{align*}

The above condition on $n$ is extremely mild (for $d=2$ it is $n > 5 \ln^2(1/\delta)$ $\ln^2(10 \ln (1/\delta))$, for $d=3$,  $n > 7 \ln^2(1/\delta) \ln^2(14 \ln(1/\delta))$, and for $d=4$, $n > 9 \ln^2(1/\delta) \ln^2(18 \ln(1/\delta)$).

For $d=2$, the induced grid error is already accounted for in the above analysis, so we can state $C_2 \leq 4 \ln(8/\eps)$ and 
$\eps(n,\delta) := \sqrt{\frac{4 \ln (2n)}{n}\ln(1/\delta)}$.  
For $d=3,4$ we could apply the same approach via Theorem \ref{thm:dKS-d}, but recall this would result in a roughly $1/\eps^3$ or $1/\eps^4$ time algorithm, which we asymptotically improve in Section \ref{sec:near-lin34}.  These asymptotic runtime improvements result in another additive error $\eps$.  Thus, we instead need to set $\eps' = \eps/3$; this induces $T = d^{d+1}3^d/\eps^d$ and ultimately $C_d = 2d \ln(d 3 d^{1/d}/\eps) \leq 2d \ln(6d /\eps)$.  For $d=3,4$ we have $C_3 \leq 6 \ln (18/\eps)$ so $\eps := \sqrt{\frac{6 \ln (3n)}{n} \ln(1/\delta)}$ 
and $C_4 \leq 8 \ln(24 / \eps)$ so $\eps := \sqrt{\frac{8 \ln (3 n)}{n} \ln(1/\delta)}$.

Bounds for $C_d$ independent of $\eps$ are possible using the mentioned result of Li, Long, and Srinivasan~\cite{LLS01} that uses chaining~\cite{talagrand2014upper} and has large asymptotic factors in the best-known bounds.  
From Csikos and Mustafa~\cite{csikos2022optimal} we derive a bound of $C_d = 2205 d$ (in Appendix \ref{app:precise}), and hence 
$\eps(n,d) \leq \sqrt{\frac{2205 \cdot d}{n} \ln(1/\delta)}$.

However, we recommend the bounds with an additional $\log n$ factor.  To illustrate why, consider the $d=2$ case, where we can now claim
\[
\eps(n,\delta) = \min \left\{ \sqrt{\frac{4 \ln(2n)}{n} \ln(1/\delta)} , \sqrt{\frac{2205\cdot2}{n} \ln(1/\delta)}  \right\}.
\]
And we only employ the second term when $2 \ln(2n) > 2205$, which requires $n > 10^{478}$.

\section{Empirical Performance for d = 2}
\label{sec:epirical}

Here we provide an empirical comparison between \ourAlgo (Algorithm \ref{alg:ouralgo}), which runs in $O(n \log(n) )$ time, versus the \baseline (Algorithm \ref{alg:baseline}), which checks all distinct dominating rectangles, efficiently updating the value as it sweeps (as in Xiao~\cite{xiao2017fast}), and runs in $O(n^2)$ time for $d=2$. Focusing on the $d=2$ case, we compare first on the task of $\dKS(\mu,\nu)$ approximation, and second on the task of hypothesis testing based on $\dKS$, where we examine statistical test power.

\subsection{Distance Approximation}
\label{sec:runtime-exp}
Since the accuracy and runtime of the methods are distribution independent, we keep things simple and set $\mu = \nu = \mathsf{Unif}([0,1]^2)$; the uniform distribution over the unit square.  Thus, we know $\dKS(\mu,\nu) =0$, and therefore the observed error is just the value $\hat v$ estimated by the algorithm.  We increase the size of the samples drawn from $\mu, \nu$ (using $n$ to denote the size of each).  In Figure \ref{fig:plots} we measure and plot the runtime and error as a function of $n$.   All run times and error measurements are shown as the average over 20 samples.

\begin{figure}
\centering
\includegraphics[width=0.49\textwidth]{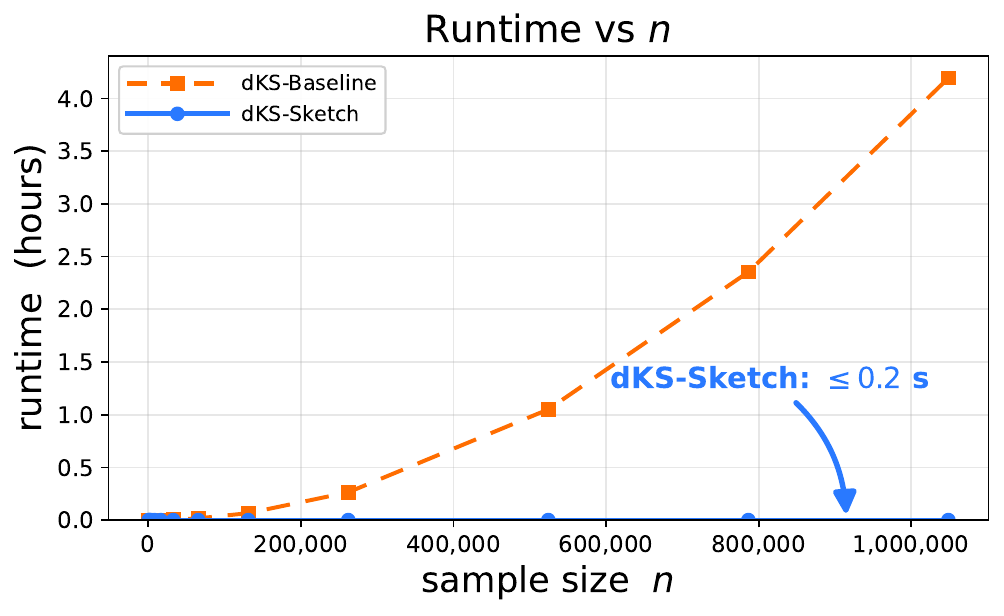}\hfill
\includegraphics[width=0.49\textwidth]{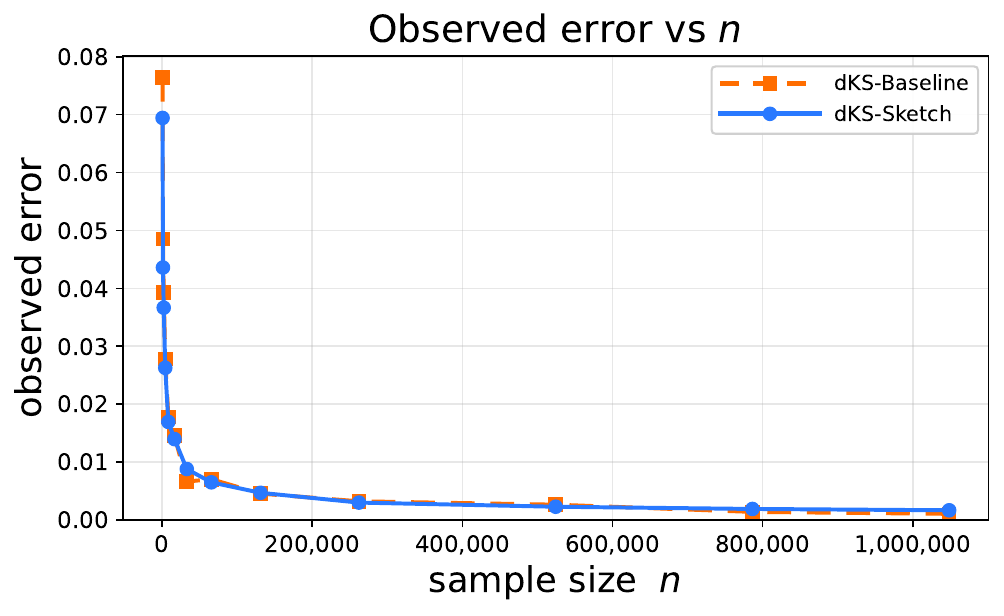}
\includegraphics[width=0.7\textwidth]{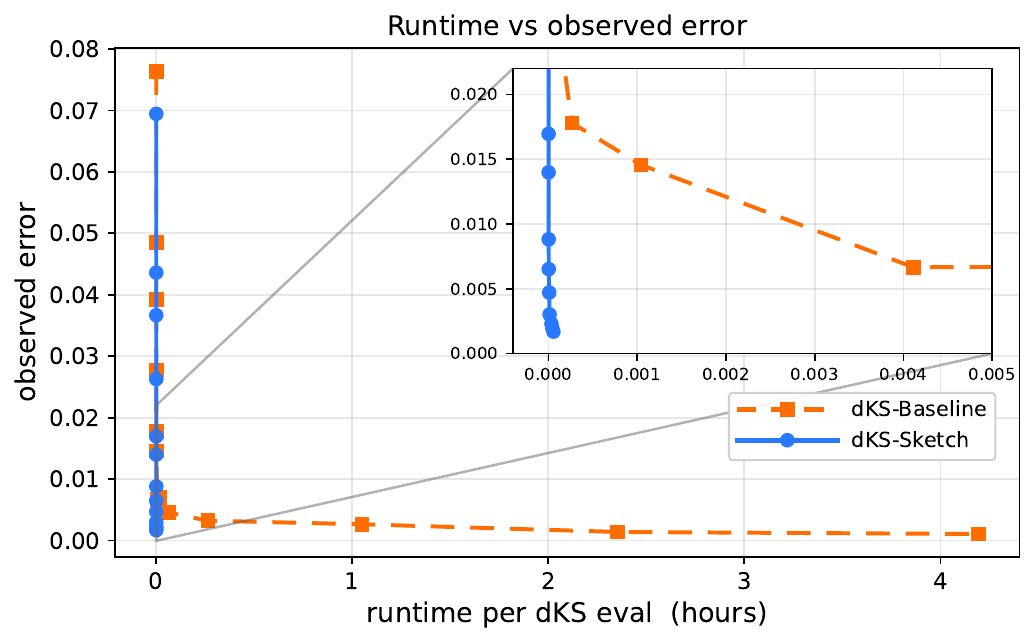} 
\caption{
Runtime and observed error vs.\ sample size $n$ for the \baseline 
(dashed orange) and \ourAlgo (solid blue) on $\mu=\nu=\mathsf{Unif}([0,1]^2)$. Since the true
$\dKS$ is $0$, the observed error is each algorithm's returned value.  
Bottom shows empirical computational-statistical runtime, with observed error vs.\ runtime.  }
\label{fig:plots}
\end{figure}

We see our \ourAlgo has runtime increasing near-linearly in $n$, while the baseline grows quadratically. At
$n{\approx}1.05$M the \baseline takes about $4.2$ hours, whereas \ourAlgo returns
in about $0.2$ seconds --- a factor of roughly $76{,}000$, and the gap widens with $n$. The two
algorithms' observed error stays close throughout and drops under $0.002$ at this $n$, so they converge
at the same rate while \ourAlgo reaches that accuracy at negligible cost.
Finally, in the bottom plot we combine this data to plot observed error as a factor of time, empirically evaluating the computational-statistical runtime.  Dramatically, \ourAlgo achieves extremely small error on a much faster time scale than the \baseline.

\subsection{Power Evaluation}
\label{sec:power-exp}
Next we evaluated the statistical power by demonstrating the efficacy of our direct precise test.  
For this, we use $\mu \neq \nu$ where still $\mu = \mathsf{Unif}([-1,1]^2)$ and $\nu$ is a mixture parameterized by $\alpha$ so $\nu_\alpha = (1-\alpha) \mathsf{Unif}([-1,1]^2) + \alpha \mathcal{N}((0,0), I/10)$.  That is $\nu_\alpha$ mixes the uniform with a small 2-dimensional normal bump at the center, with identity covariance $I$ divided by $10$.  
We fix $\delta = 0.05$ and vary the sample size $n$ (for each contamination level $\alpha$). 
Then sample $P = (X_1,\dots,X_n)$ and $Q = (Y_1,\dots,Y_n)$ where $X_1,\dots,X_n \overset{iid}{\sim} \mu$ and $Y_1,\dots,Y_n \overset{iid}{\sim} \nu_{\alpha}$, and reject the null if we find $\text{\ourAlgo}(P,Q) > \eps(n,\delta) = \sqrt{4 \ln(2n) \ln(1/\delta)/n}$.

\begin{figure}[t]
\centering
\includegraphics[width=0.49\textwidth]{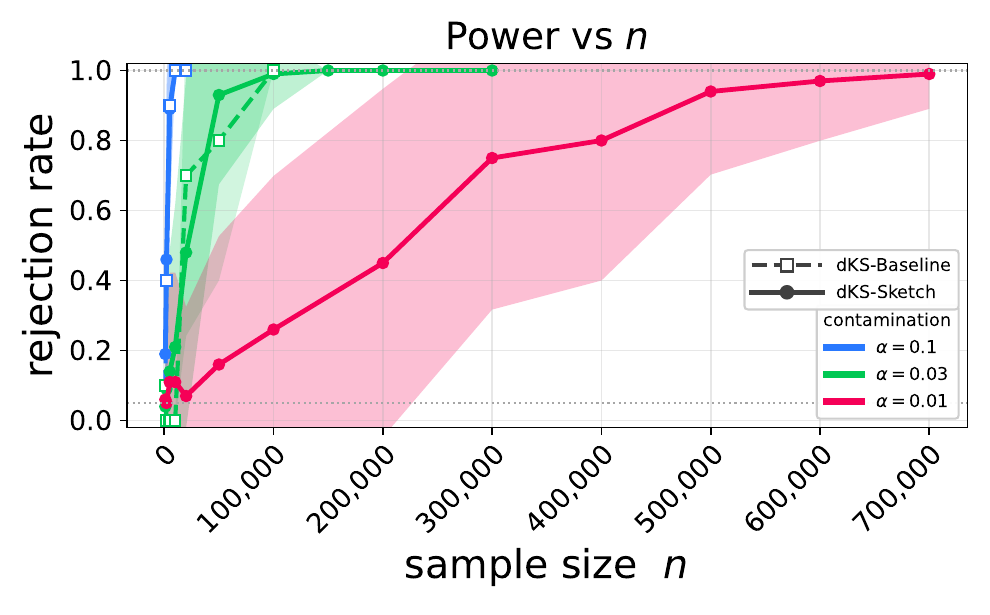}
\includegraphics[width=0.49\textwidth]{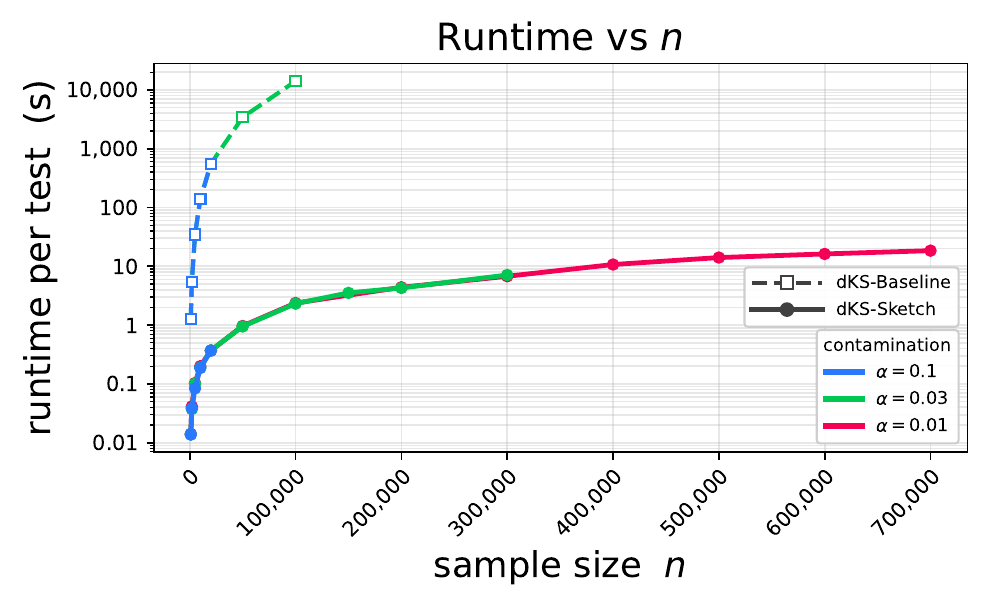}
\includegraphics[width=0.7\textwidth]{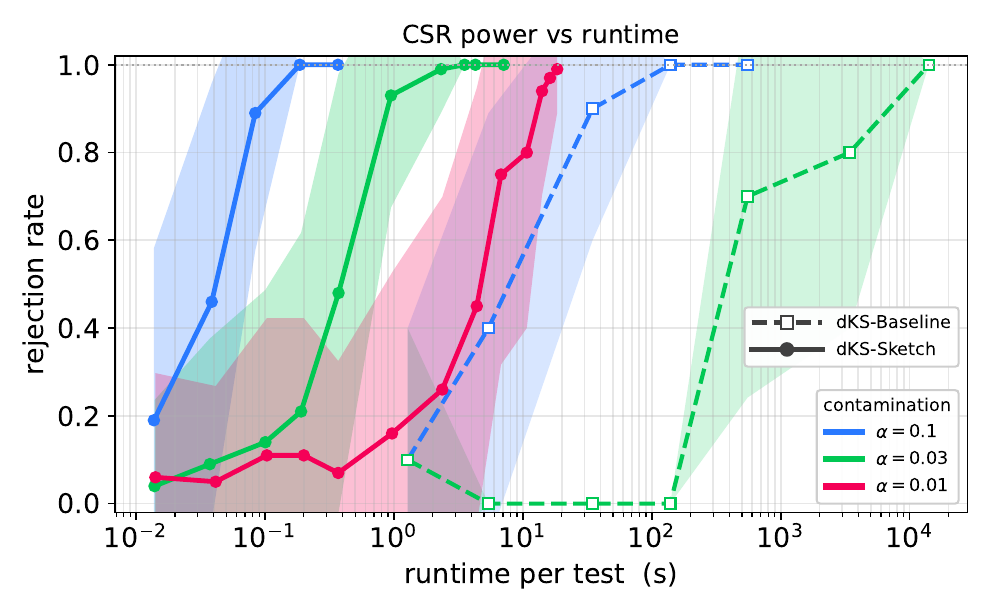}
\caption{Power and runtime of the \ourAlgo (solid, circles) vs.\ Monte Carlo exact test for \baseline with $B=100$ permutations
(dashed, squares).  The plots show runtime in seconds on a log-scale and rejection rate (with level $\delta = 0.05$) for the two-sample test, for contamination $\alpha\in\{0.1,0.03,0.01\}$ (encoded by color).}
\label{fig:power}
\end{figure}

As a comparison for the power, we consider a permutation exact test that builds off of our baseline.  We leverage the permutation test with $B=100$ 
from Proposition 3 in Hemerik and Goeman~\cite{hemerik2018exact}.  For samples $P, Q$ (each of size $n$ with $P$ an iid sample from $\mu$ and $Q$ an iid sample from $\nu_{\alpha}$) we concatenate them into $Z = [P; Q]$ representing $2n$ points in $\R^2$.  We then sample $B-1$ permutations of $Z$ as $Z_1, Z_2, \ldots, Z_{B-1}$.  We compute $v'_j$ for each in two steps:   
\\ (1) splitting $Z_j$ into the first $n$ points as $P_j$ and the last $n$ as $Q_j$ and 
\\ (2) running $\text{\baseline}(P_j,Q_j) \to v'_j$.  
\\ Also compute $v' = \text{\baseline}(P,Q)$ and create the set $S = \{v', v'_1, \ldots, v'_{B-1}\}$.  We set threshold $\tau$ as the $\lceil (1-\delta) B \rceil$th largest value in $S$.  If $v' = \text{\baseline}(P,Q)$ is larger than $\tau$ we reject, if its less than $\tau$ we fail to reject.  If $v' = \tau$ then we reject with probability $(\delta B - M_1)/M_0$ where $M_1$ is the number of $v \in S$ strictly larger than $\tau$, and $M_0$ is the number of $v \in S$ equal to $\tau$.

Figure \ref{fig:power} shows the power results for these two methods.  For the direct precise test with \ourAlgo we set 
$\alpha\in\{0.1,0.03,0.01\}$, and $W=100$ trials to average over.  
For the Monte Carlo \baseline approach for computational constraints we only consider $\alpha = \{0.1, 0.03\}$, and $W=10$ repetitions.  For each we report the mean and show $\pm$ 1 standard deviation.  
In the top two plots we observe that the rejection rate scales about the same for the two methods as a function of $n$.  However \ourAlgo is many orders of magnitude faster that \baseline for the same $n$.  Unlike the experiments in Section \ref{sec:runtime-exp}, this now comes from two sources: the algorithmic improvement $O(n^2)$ to $O(n \log n)$, and the factor $B=100$ Monte Carlo trials.  While our analysis could be used to derive a direct precise test based on \baseline, we are unaware of one in the literature previously, and wanted to highlight this large discrepancy in computational cost.  
This issue is shown more directly in the bottom plot which shows the computational-statistical runtime, with average runtime on the $x$-axis and rejection rate on the $y$-axis.  
At $\alpha=0.1$, \ourAlgo is fully powered in less than a second, where as \baseline takes about $100$ seconds; this is mostly attributed to the Monte Carlo factor $B=100$.  However at $\alpha = 0.03$ while \ourAlgo achieves full power at under $3$ seconds, \baseline takes about $10,000$ seconds (about 3 hours).  This difference is due to both factors -- runtime improvement and precise vs. Monte Carlo test.  Then for $\alpha=0.01$ full power requires about $n = 7 \times 10^5$.  At this scale \ourAlgo takes about 18 seconds, whereas \baseline would need about 8 days per replication (we did not attempt to run this $\approx$ 80 day experiment).  
Altogether, \ourAlgo offers a dramatic improvement in computational-statistical runtime.

\section{Discussion on Future Directions}
\label{sec:discussion}
The core of our implementation is C++, and also available with Python front end:  \url{https://github.com/foadnamjoo/dKS}.

The described algorithms for $d=3,4$, which are near-linear in $n$, do not appear to be simple to implement since they rely on a reduction to a complex recursive procedure for Klee's measure problem~\cite{chan2013klee} -- despite this approach having been ``made easy'' in comparison to prior solutions.  Yet, we are hopeful that there are even simpler and more practical solutions for $d > 2$ that also achieve near-linear runtime in $n$.  In particular, for $d=3$, the Klee's measure part only takes $\tilde O(1/\eps^{3/2})$ time, so there should be room to trade off theoretical runtime for practical improvements.  

Moreover, we believe that our framework of \emph{computational-statistical runtime} for statistical measures, where we assume that our only access to the input is $O(1)$-time samples, is relatively unexplored.  We expect that future work will develop this setting for other statistical measures.  In particular, the approach we take three natural steps: (1) draw $n_\eps$ samples, (2) sketch the data into a sparser representation, and (3) solve for the distance.  These appear to be a general paradigm for future study in these kinds of problems.

\bibliography{coreset.bib}

\appendix

\section{Measure Theoretic Metric Properties}
\label{app:metric}

To prove that $\dKS$ is a metric on measures $\mu,\nu$, we will need to define and deploy some technical concepts from measure theory.  Here we  first remind the reader of relevant language and notation. Recall that a collection of subsets $S$ of a ground set $\mathcal{X}$ is called a sigma-algebra if $S$ contains $\mathcal{X}$ and is closed under complements and countable unions. Moreover, regardless of whether or not $S$ is a sigma-algebra, the smallest sigma-algebra containing $S$ is called the sigma-algebra \textit{generated} by $S$. Moreover, $S$ is called a semi-algebra if it contains the empty set and $\mathcal{X}$, is closed under finite intersections, and the complement of every set in $S$ is expressible as a finite union of other sets in $S$. $\mathcal{B}({\mathbb{R}^{d}})$ is the notation for the smallest sigma algebra of subsets of $\mathbb{R}^{d}$ generated by the open sets of $\mathbb{R}^{d}$; these are the Borel sets of $\mathbb{R}^{d}$. This set of sets contains the reasonable sets any probability measure on $\mathbb{R}^{d}$ should be able to assign measure to.

We will also need a small amount of additional notation for the differences in the ranges we consider, and how we define sets, especially to handle unbounded ranges.  It will be useful to prove some results for traditional rectangles $\TT_d$ and map back to dominating rectangles $\RR_d$.  
For $z',z \in \bar{\mathbb{R}}^{d}$ where $\bar{\mathbb{R}}$ is the extended real line, $z' \preceq_{\infty} z$ means whenever the $j$th coordinate of $z$ is finite, $z'_j \leq z_j$, and when the $j$th coordinate of $z$ is infinity, $z'_j < z_j$. Earlier we defined the axis aligned rectangles in $\mathbb{R}^{d}$ as $\TT_d := \{R_{z,z'} \mid z,z' \in \mathbb{R}^{d},z' \preceq z \}$ where $R_{z,z'} = \{x \in \mathbb{R}^{d} \mid z' \prec x \preceq z\}$. Now we define a larger class of rectangles consisting of bounded and unbounded rectangles. Specifically, for $z,z' \in \bar{\mathbb{R}}^{d}$, let $R_{z,z'}^{\infty} := \{ x \in \mathbb{R}^{d} \mid z' \prec x \preceq_{\infty} z \}$, and define $\bTT_d := \{R_{z,z'}^{\infty} \mid z,z' \in \bar{\mathbb{R}}^{d},z' \preceq z\}$. We start with a lemma that shows probability measures agreeing on the dominating rectangles also agree on $\bar{\TT_d}$.
\begin{lemma}
\label{lem:rect-agree}
If $\mu,\nu$ are probability measures on $(\mathbb{R}^{d},\mathcal{B}(\mathbb{R}^{d}))$ agreeing on $\RR_d$, then $\mu,\nu$ agree on $\bar{\TT_d}$.
\end{lemma}
\begin{proof}
    Consider first a rectangle $R_{z,z'}^{\infty} \in \bTT_d$ that is bounded. That is, $z,z' \in \mathbb{R}^{d}$. In particular, $R_{z,z'}^{\infty} = R_{z,z'}$. Denote by $\mathcal{V}(R_{z,z'})$ the $2^d$ corners of $R_{z,z'}$. Then we first claim that 
    \[
    R_{z} = R_{z,z'} \cup \left(\bigcup_{v \in \mathcal{V}(R_{z,z'}) \setminus \{z\} }  R_{v} \right).
    \]
    To see this, we will prove the two set containments. For reverse containment, first note that since $z' \preceq z$, for $v \in \mathcal{V}(R_{z,z'})$, $v \preceq z$. In particular, $R_{v} \subseteq R_z$. Also, it is immediate that $R_{z,z'} \subseteq R_{z}$, so we conclude reverse containment. For forward containment, suppose $x \in R_{z}$. Then since $R_{z,z'} \subseteq R_{z}$, either $x \in R_{z,z'}$ or not. And if not, we know that both $x \preceq z$ and it is not the case that $z' \prec x$. In particular, there exists a non-empty set of indices $S(x) \subseteq [d]$ such that $x_j \leq z'_j$ for $j \in S(x)$. And for all indices $j \in [d]\setminus S(x)$, $x_j \leq z_j$. In particular, letting $v$ be the vertex with value $z'_j$ at indices $j \in S(x)$ and $z_j$ at indices $j \in [d] \setminus S(x)$, we have that $x \preceq v$. That is, $x \in R_{v}$, and since $S(x)$ is non-empty, $v \in \mathcal{V}(R_{z,z'}) \setminus \{z\}$, so $x \in \left(\bigcup_{v \in \mathcal{V}(R_{z,z'}) \setminus \{z\} }  R_{v} \right)$. We have thus shown if $x \in R_{z}$, then either $x \in R_{z,z'}$ or $x \in \left(\bigcup_{v \in \mathcal{V}(R_{z,z'}) \setminus \{z\} }  R_{v} \right)$, so containment holds.

    Now recall that $\RR_d$ is closed under finite intersections, and the inclusion-exclusion formula (see for example \cite{resnick2013probability}) guarantees that for a probability measure, the measure of a finite union is precisely a finite sum of probability measures of finite intersections of members of the union. Therefore, since $\mu$ and $\nu$ agree on $\RR_d$, they must agree on $\bigcup_{v \in \mathcal{V}(R_{z,z'}) \setminus \{z\} }  R_{v}$. Finally, note that $R_{z,z'}$ is disjoint from $\bigcup_{v \in \mathcal{V}(R_{z,z'}) \setminus \{z\} }  R_{v}$. Using this, the set equality, and finite additivity, we conclude $\mu$ and $\nu$ agree on $R_{z,z'} = R_{z,z'}^{\infty}$.

    All that is left to show is that $\mu,\nu$ agree on unbounded rectangles as well. But every unbounded rectangle is a countable union of disjoint bounded rectangles. So by countable additivity and since we have just shown $\mu,\nu$ agree on all bounded rectangles, they must also agree on all unbounded rectangles. In particular, $\mu,\nu$ agree on $\bar{\TT_d}$.
\end{proof}

Now we can prove our main result about metric properties of $\dKS$.  

\begin{theorem}[Restated Theorem \ref{thm:dKS-metric}]
    \label{thm:dKS-metric-app}
    On the collection of probability measures defined on $(\R^d,\mathcal{B}(\R^d))$, $\dgen_{\RR_d} = \dKS$ is a metric.
 \end{theorem}
 \begin{proof}
     Because $\dKS$ is in the form of an IPM, it is at least a pseudometric, and what remains is to show it satisfies the strong form of identity: that if $\mu,\nu$ are probability measures on $(\mathbb{R}^d,\mathcal{B}(\mathbb{R}^{d}))$, then $\dKS(\mu,\nu) =0$ implies that $\mu = \nu$. So suppose $\dKS(\mu,\nu) =0$. We will use a measure theory extension theorem, Theorem 2.4.3 of \cite{resnick2013probability}. Consider a set function defined on a semi-algebra of subsets of $\mathbb{R}^{d}$ that assigns measure $1$ to $\mathbb{R}^{d}$ and is countably additive on the semi-algebra; this theorem states that this set function has a unique extension to the sigma algebra generated by the semi-algebra. To use the theorem we will choose a semi-algebra that generates $\mathcal{B}(\mathbb{R}^d)$. It is classically known (\cite{durrett2019probability} Chapter 1, section 1) that $\bar{\TT_d}$ is a semi-algebra, and the sigma algebra generated by $\bar{\TT_d}$ is $\mathcal{B}(\mathbb{R}^{d})$. Also $\mu,\nu$ are probability measures on $\mathcal{B}(\mathbb{R}^{d})$ so they are countably additive on any collection of subsets, including $\bar{\TT_d}$. Finally, since by assumption $\mu,\nu$ agree on $\RR_d$, by Lemma \ref{lem:rect-agree} $\mu,\nu$ also agree on $\bar{\TT_d}$.  So by Theorem 2.4.3 of \cite{resnick2013probability}, $\mu$ and $\nu$ agree on $\mathcal{B}(\mathbb{R}^{d})$.
 \end{proof}

\section{Asymptotically Optimal Precision Testing Bounds}
\label{app:precise}

In this section we derive an absolute constant $C_d$ to be used in the error threshold $\eps(n,\delta)$ for $\delta$-precision testing with $\dKS$, in the formula
\[
\eps(n,\delta) = \sqrt{\frac{C_d \log(1/\delta)}{n}}.
\]

Bounds for $C_d$ can be derived via sample complexity bounds for approximating a range space with bounded combinatorial complexity~\cite{VC71}.  Bounds for $C_d$ independent of $\eps$ in this context first appeared in Li, Long, and Srinivasan~\cite{LLS01} using Talagrand's chaining~\cite{talagrand2014upper} technique.  However, these have large asymptotic factors, and while some constants are provided, they are not easily pieced together to get the desired bound on $C_d$.  

More recently Csikos and Mustafa~\cite{csikos2022optimal} provided an alternate derivation of such a bound from which we can derive a constant $C_d = 2205 d$.  This starts with a stronger $(\alpha,\eps)$-approximation of a range space $(X,\R_d)$ with a sample $S$ of size $n$ so that for all $R \in \RR_d$ we have
\[
\left| \frac{|R \cap X|}{|X|} - \frac{|R \cap S|}{n}   \right| \leq \alpha \max \left\{ \frac{|R \cap X|}{|X|},  \beta  \right\}.
\]
Their bound requires $0< \alpha, \beta, \gamma \leq 1/2$ and 
\[
n \geq c_2 \frac{d}{\alpha \beta} \left(\frac{1}{\beta} \ln \frac{1}{\alpha} + \ln \frac{e t}{d \gamma}\right)
\]
and succeeds with probability at least $1-\gamma$, where 
\[
t = c_1 \frac{1}{\alpha \beta^2} \left(d \ln(\frac{1}{\alpha \beta}) + \ln(1/\gamma)\right).
\]
Applying this in our context, we require $\beta = (2/3) \eps$, $\gamma = \delta/2$, and we can set $\alpha = 1/2$ so (the LHS above) $\alpha \beta \leq \eps/3$.  They report this works for constants $c_1 = 318$ and $c_2 = 70$.  

First we can derive 
\begin{align*}
  t 
  & = 
  c_1 \frac{1}{\alpha \beta^2} \left(d \ln(\frac{1}{\alpha \beta}) + \ln(1/\gamma)\right)
  \\ & = 
  318 \frac{1}{(1/2) ((2/3)\eps)^2} \left(d \ln \left(\frac{1}{(1/2) ((2/3)\eps)}\right) + \ln(2/\delta)\right)
  \\ & = 
  \frac{1431}{\eps^2} \left(d \ln(3/\eps) + \ln(2/\delta)\right)
  \\ & \geq 
  \frac{1431 d}{\eps^2} \left(\ln(3/\eps) + \ln(2/\delta)\right)\end{align*}
This leads to a bound of 
\begin{align*}
  n 
  &\geq c_2 
  \frac{d}{\alpha \beta} \left(\frac{1}{\beta} \ln \frac{1}{\alpha} + \ln \frac{e t}{d \gamma}\right)
  \\ &=
  70 \frac{d}{\eps/3} \left(\frac{3}{2\eps} \ln 2  + \ln \left(\frac{2 e}{\delta}\frac{t}{d} \right)\right)
  \\ & \geq
  210 \frac{d}{\eps} \left(\frac{3}{2\eps} + \ln \left(\frac{2 e }{\delta} \frac{1431}{\eps^2}(\ln(3/\eps) + \ln(2/\delta) \right)\right)
  \\ & \geq
  210 \frac{d}{\eps} \left(\frac{3}{2\eps} + \ln \left(\frac{7780}{\delta^2 \eps^3} \right)\right)
  \\ & \geq
  210 \frac{d}{\eps} \left(\frac{3}{2\eps} + 3 \ln(20/\eps \delta) \right)
  \\ & \geq
  210 \frac{d}{\eps} \left(\frac{3}{2\eps} + \frac{9}{\eps} \ln(1/\delta) \right)
  \\ & \geq
  210 \frac{d}{\eps} \left(\frac{21}{2\eps}\ln(1/\delta) \right)
  \\ & \geq
  2205 \frac{d}{\eps^2} \ln(1/\delta).  
\end{align*}

\end{document}